\documentclass[ejs]{imsart}
\RequirePackage[OT1]{fontenc}
\usepackage{amsmath,amsthm, amssymb,graphicx,color,url,float,subfig,enumerate}
\RequirePackage[numbers]{natbib}

\doi{10.1214/154957804100000000}
\pubyear{0000}
\volume{0}
\firstpage{0}
\lastpage{0}

\startlocaldefs
\let\hat\widehat
\let\tilde\widetilde

\newcommand{\X}{{\mathcal X}}
\newcommand{\R}{{\mathbb R}}
\newcommand{\E}{{\mathbb E}}
\newcommand{\V}{{\mathbb V}}
\newcommand{\x}{{\bf x}}
\newcommand{\comment}[1]{}

\newcommand{\CP}[2]{ \mathbb{P} \left( \left. {#1} \,\right| {#2} \right)}
\newcommand{\CV}[2]{ \V \left( \left. {#1} \,\right| {#2} \right)}
\newtheorem{theorem}{Theorem}
\newtheorem{definition}[theorem]{Definition}
\newtheorem{lemma}[theorem]{Lemma}
\newtheorem{example}[theorem]{Example}
\newtheorem{corollary}[theorem]{Corollary}
\newtheorem{proposition}[theorem]{Proposition}
\newtheorem{remark}[theorem]{Remark}

\endlocaldefs

\newcommand{\red}[1]{\textbf{\color{red} ***ATT*** #1}}

\begin{document}

\begin{frontmatter}

\title{A Spectral Series Approach to High-Dimensional Nonparametric Regression}
\runtitle{Nonparametric Regression via Spectral Series}


\begin{aug}
  \author{Ann B. Lee\thanksref{t2}\ead[label=e1]{annlee@cmu.edu}}

  \address{Department of Statistics, Carnegie Mellon University, USA.\\ 
           \printead{e1}}

        \author{Rafael Izbicki
        	\ead[label=e2]{rafaelizbicki@gmail.com}}

\address{Department of Statistics, Federal University of S\~ao Carlos, Brazil.\\ 
	\printead{e2}}

  \thankstext{t2}{Corresponding author}

\runauthor{Lee and Izbicki}
\end{aug}

\begin{abstract} A key question in modern statistics is how to make fast and reliable inferences for complex, high-dimensional data. While there has been much interest in sparse techniques, current methods do not generalize well to data with nonlinear structure. In this work, we present an orthogonal series estimator for predictors that are complex aggregate objects, such as natural images, galaxy spectra, trajectories, and movies. Our series approach ties together ideas from 
 manifold learning, kernel machine learning, and Fourier methods. We expand the unknown regression on the data in terms of the eigenfunctions of a kernel-based operator, and we take advantage of orthogonality of the basis with respect to the underlying data distribution, $P$, to speed up computations and  tuning of parameters.  If the kernel is appropriately chosen, then the eigenfunctions adapt to the intrinsic geometry and dimension of the data. We provide theoretical guarantees for a radial kernel with varying bandwidth,  and we relate smoothness of the regression function with respect  to $P$ to sparsity in the eigenbasis.  Finally, using simulated and real-world data, we systematically compare the performance of the spectral series approach with classical kernel smoothing, k-nearest neighbors regression, kernel ridge regression, and state-of-the-art manifold and local regression methods.
\end{abstract}

\comment{
\begin{abstract}
A key question in modern statistics is how to make efficient inferences for complex, high-dimensional data. While there has been much interest in sparse techniques, current methods do not generalize well to data with near collinearity in the predictors. In this work, we propose a nonparametric orthogonal series estimator where the predictors can be non-standard data objects, e.g., images, spectra, trajectories, and movies. Our approach provides a unifying framework for manifold learning, kernel machine learning, and Fourier methods. The basic idea behind the series method is to expand the regression function in terms of the eigenfunctions of a kernel-based operator. 
These eigenfunctions are orthogonal with respect to the underlying data distribution, $P$, as opposed to the Lebesgue measure of the ambient space. If the kernel is appropriately chosen, then the eigenfunctions adapt to the intrinsic geometry and dimensionality of the data. We provide theoretical guarantees on the estimator and we relate smoothness of the regression function relative $P$ to sparsity in the eigenbasis. 
 We also discuss the computational advantages of the series method and demonstrate its effectiveness for Isomap face images and galaxy spectra from the Sloan Digital Sky Survey.
\end{abstract}
}

\begin{keyword}[class=MSC]
\kwd{62G08}  
\end{keyword}

\begin{keyword}
\kwd{high-dimensional inference, orthogonal series regression, data-driven basis, Mercer kernel, manifold learning, eigenmaps}
\end{keyword}



\end{frontmatter}

\section{Introduction} 
A challenging problem in modern statistics is how to handle complex, high-dimensional data. Sparsity has emerged as a major tool for making efficient inferences and predictions for multidimensional data. Generally speaking, sparsity refers to a situation where the data, despite their apparent high dimensionality, are highly redundant with a low intrinsic dimensionality. 
In our paper, we use the term ``sparse structure''  to refer to cases where the underlying distribution $P$
places most of its mass on a subset $\mathcal{X}$ of $\mathbb{R}^d$ of small Lebesgue measure. This scenario includes, but is not limited to, 
 Riemannian submanifolds of $\mathbb{R}^d$, and high-density clusters separated by low-density regions. In applications of interest, observable data often have (complex) sparse structure
due to the nature of the underlying physical systems. For example, in astronomy, raw galaxy spectra are of dimension
equal to the number of wavelength measurements $d$, 
 but inspection of a sample of such spectra will reveal clear, low-dimensional
 features and structure resulting from the shared physical system that generated these galaxies.
 While the real dimensionality of data 
 is much smaller than $d$, the challenge remains to exploit this when predicting, for example, the age, composition, and star formation history of a galaxy.

In its simplest form, low-dimensional structure is apparent in the original coordinate system. Indeed, in regression, much research on ``large p, small n'' problems concerns {\em variable selection} and the problem of recovering a ``sparse'' coefficient vector (i.e., a vector with mostly zeros) with respect to the given variables. Such approaches include, for example, lasso-type regularization~\citep{Tibshirani:96}, the Dantzig selector~\citep{Candes_Tao}, and RODEO~\citep{Lafferty:Wasserman:2008}. There are also various extensions that incorporate lower-order interactions and groupings of covariates \citep{Yuan:Lin:2006, Zou:Hastie:05, Ravikumar:EtAl:2009} but, like lasso-type estimators, they are not directly applicable to the more intricate structures observed in, e.g.,  natural images, spectra, and hurricane tracks.

At the same time, there has been a growing interest in statistical methods that explicitly consider {\em geometric} structure in the data themselves. Most traditional dimension-reducing regression techniques, e.g., principal component regression (PCR;~\cite{Jolliffe:2002})
partial least squares (PLS;~\cite{Wold:EtAl:2001}) and sparse coding~\citep{Olshausen+-1996}, are based on linear data transformations and enforce sparsity (with respect to the $L^1$ or $L^2$ norm) of the regression in a rotated space. 
 More recently, several authors \citep{Bickel:Li:2007, Aswani:EtAl:2011, cheng2012local} 
 have studied local polynomial regression methods on {\em non-linear} manifolds. For example,  Aswani et al.~\citep{Aswani:EtAl:2011}  
 propose a geometry-based regularization scheme 
that estimates the local covariance matrix at a point and then penalizes regression coefficients perpendicular to the estimated manifold direction.  In the same spirit, Cheng and Wu \cite{cheng2012local} suggest first reducing the dimensionality to the estimated intrinsic dimension of the manifold, and then applying local linear regression to a tangent plane estimate. Local regression and manifold-based methods tend to perform well when there is a clear submanifold  but these approaches are not practical in higher dimensions or when the local dimension varies from point to point in the sample space. Hence, existing nonparametric models still suffer when estimating unknown functions (e.g., density and regression functions) on complex objects $\x \in \mathcal{X} \subset \Re^d$, where $d$ is large. 

Much statistical research has revolved around adapting classical methods, such as linear, kernel-weighted, and additive models to high dimensions. 
 On the other hand, statisticians have paid little attention to the potential of {\em orthogonal series} approaches. 
 In low dimensions, orthogonal series is a powerful nonparametric technique
for estimating densities and regression functions. Such methods are fast to implement with easily interpretable results, they have sharp 
optimality properties, and a wide variety of bases allows the data analyst to model multiscale structure and any challenging shape of the target function~\cite{efromovich}. 
As a result, Fourier series approaches have dominated 
research in signal processing and mathematical physics. 
  This success, however, has not translated to more powerful nonparametric tools in dimensions of the order of $d \sim 100$ or $1000$; 
in fact, extensions via tensor products (as well as more sophisticated adaptive grid or 
triangulation methods; see~\cite{Mallat:2009} and references within) quickly 
become unpractical in dimensions $d>3$. 

In this work, we will build on ideas from harmonic analysis and spectral methods to
construct nonparametric methods
for estimating unknown functions in high-dimensional spaces with non-standard data
objects (such as images, spectra, and distributions) that possess {\em sparse nonlinear structure}.
We derive a Fourier-like basis
$\left\{\psi_j(\x)\right\}_{j \in \mathbb{N}}$ of $L^2(\mathcal{X})$ that adapts
to the intrinsic geometry of the underlying data distribution $P$, and which is orthonormal with respect to $P$ rather than the 
Lebesgue measure of the ambient space.  The empirical basis
functions are then used to estimate functions  
on complex data $\x \in \mathcal{X}$; 
such as, for example, the regression function $r(\x)=\mathbb{E}(Y|\mathbf{X}=\x)$ of a response variable $Y$ on an object $\x$.
Because of the adaptiveness of the basis, there is no need for high-dimensional tensor products. 
Moreover, we take advantage of the orthogonality property of the basis for fast computation 
and model selection.  
We refer to our approach as {\em spectral series} as it is based on spectral methods (in particular, diffusion maps~\citep{pnas1, Coifman:Lafon:06, LafonLee2006} and spectral connectivity analysis~\citep{LeeWasserman:2010}) and Fourier series. Sections~\ref{sec::setup}-\ref{sec::construction_basis} describe the main idea of the series method in a regression setting.

\comment{
In this work, we investigate whether one can 
 find {\em an orthogonal basis adapted to the underlying geometry of a high-dimensional data distribution}.
 We describe and analyze the performance of one such adaptive series approach. The method is based on a spectral analysis of the data and is closely related to spectral clustering~\citep{NgJordanWeiss01, Shi:EtAl:09} and eigenmaps; in particular, diffusion maps~\citep{pnas1, Coifman:Lafon:06} and spectral connectivity analysis~\citep{LeeWasserman:2010}.  Hence, we use the term {\em spectral series}. 
 }
 
  Our work generalizes and ties together ideas in  classical smoothing, kernel machine learning~\citep{scholkopf1997kernel, Scholkopf:Smola:2001, cucker2007learning}, support vector machines~(SVMs;~\cite{Steinwart:Christmann:2008}) and manifold regularization~\citep{Belkin:EtAl:06}  
 {\em without}  
 the many restrictive assumptions (fixed kernel, exact manifold, infinite unlabeled data and so on) seen in other works.
 There is a large literature on SVMs and kernel machine learning that use similar approximation spaces as us, 
  but it is unclear whether and how those procedures  adapt to the structure of the data distribution. Generally, there is a discrepancy between theoretical work on SVMs, which assume a {\em fixed} RKHS (e.g., a fixed kernel bandwidth), and applied SVM work, where the RKHS is chosen in a data-dependent way (by, e.g.,   decreasing the kernel bandwidth $\varepsilon_n$ for larger sample sizes $n$). 
  Indeed, issues concerning the choice of tuning parameters, and their relation to the data distribution $P$, are considered to be open problems in the mainstream RKHS literature. 
The manifold regularization work by Belkin et al.~\cite{Belkin:EtAl:06} 
   addresses adaptivity to sparse structure but  under restrictive assumptions, such as the existence of a well-defined submanifold and the presence of infinite unlabeled data. 
 
 Another key difference between our work and kernel machine learning is that we {\em explicitly} compute the eigenfunctions of a kernel-based operator and then use an orthogonal series approach to nonparametric curve estimation. 
 Neither SVMs nor manifold regularizers exploit {\em orthogonality relative to $P$}. 
 In our paper, we point out the advantages of an orthogonal series approach 
   in terms of 
 computational efficiency (such as fast cross-validation and tuning of parameters),  visualization, and interpretation. SVMs can sometimes have a ``black box feel,'' whereas the spectral series  method allows the user to directly link the data-driven Fourier-like eigenfunctions to the function of interest and the sample space. 
  Indeed, there is a dual interpretation of the computed eigenfunctions:  (i) They define {\em new coordinates of the data} which are useful for nonlinear dimensionality reduction, manifold learning, and 
 data visualization. 
  (ii) They form an orthogonal Hilbert basis for functions on the data and are a means to {\em nonparametric curve estimation} via the classical orthogonal series method, even when there is no clearly defined manifold structure. There is a large body of work in the machine learning literature addressing the first perspective;  see, e.g., Laplacian
maps~\citep{BelkinNiyogi03}, Hessian maps~\citep{DonohoGrimes03}, diffusion maps~\citep{pnas1}, Euclidean Commute Time maps~\citep{Saerens:EtAl:2004}, and  spectral clustering~\citep{Shi:EtAl:09}. In this paper, we are mainly concerned with the second view, i.e., that of estimating unknown functions on complex data objects and understanding the statistical properties of such estimators. 

Fig.~\ref{fig::diffMapISO}, for example,  shows a 2D visualization of the Isomap face data using the eigenvectors of a renormalized Gaussian kernel as coordinates (Eq.~\ref{eq:eigenmap}). Assume we want to estimate the pose $Y$ of the faces. How does one solve a regression problem where the predictor $\x$ is an {\em entire} image? Traditional methods do not cope well with this task while our spectral series approach (Eq.~\ref{eq::series_est} with estimated eigenfunctions as a basis) can use complex {\em aggregate} objects $\x$ (e.g., images, spectra, trajectories, and text data) as predictors, without an explicit dimension reduction step. 
Note that the eigenvectors capture the pose $Y$ and other continuous variations of an image $\x$ fairly well, 
and that the regression $\mathbb{E}(Y|\x)$ appears to vary smoothly in sample space. We will return to the face pose estimation problem in Sec.~\ref{sec:isomap}. We will also discuss the theoretical properties of a spectral series estimator of the regression function  $f(\x)=\mathbb{E}(Y|\x)$ in Sec.~\ref{sec::theory}, including the connection between smoothness and efficient estimators.
\begin{figure}[H]
  \centering
  \includegraphics[width=3.5in,angle=-0]{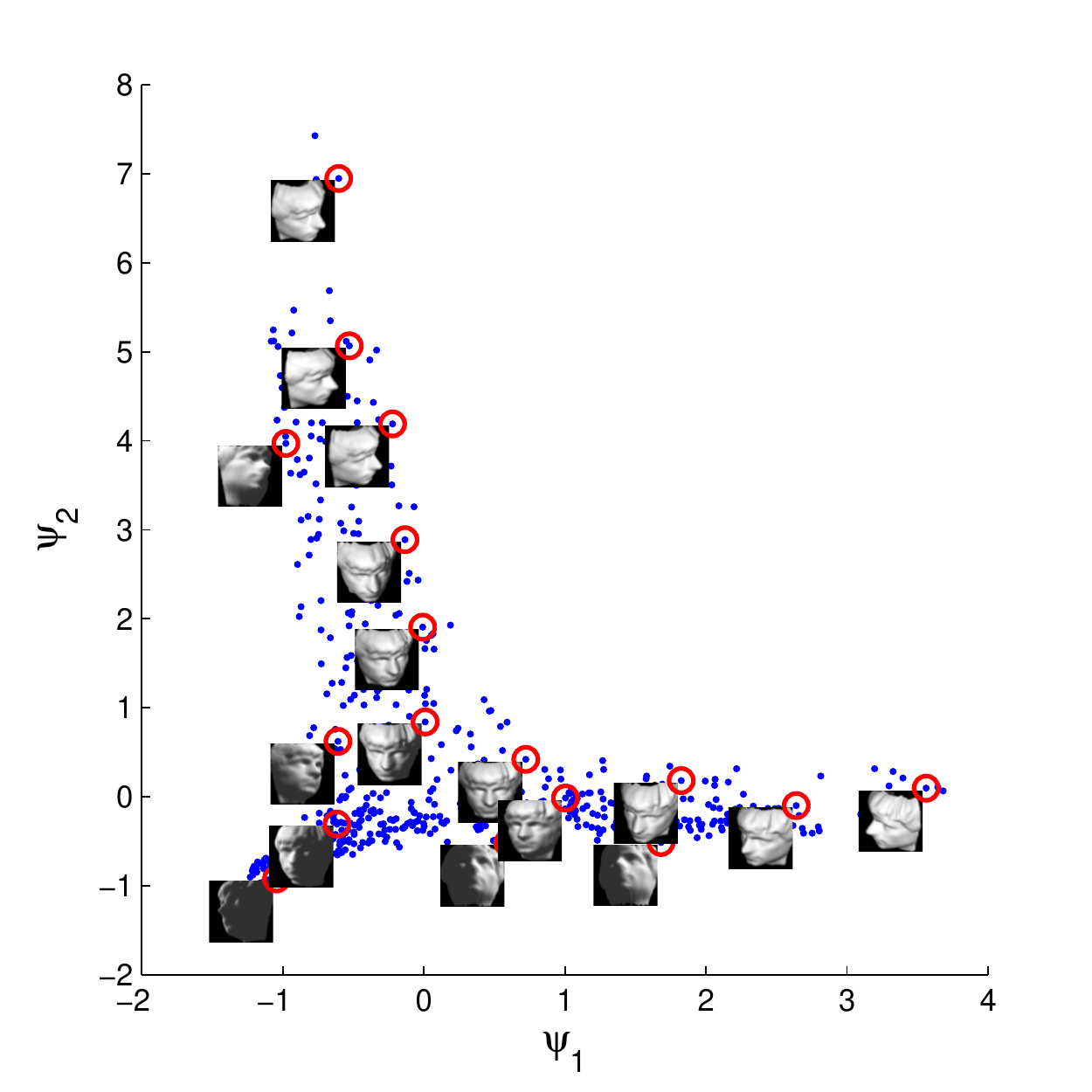} 
\vspace{-0.1in}
  \caption{ \footnotesize Embedding or so-called ``eigenmap'' of the Isomap face data using the first two non-trivial eigenvectors of the Gaussian diffusion kernel. The eigenvectors capture the pose $Y$ and other continuous variations of an image $\x$ fairly well,  
   and the regression $\mathbb{E}(Y|\x)$ appears to vary smoothly in sample space. 
  } 
   \label{fig::diffMapISO}
\end{figure}

Our paper has the following aims:
\begin{enumerate}[(i)]
\item { \em Unifying.} To generalize and connect ideas in kernel machine learning, manifold learning, spectral methods and classical 
 smoothing, without the many restrictive assumptions (fixed kernel, exact manifold, infinite unlabeled data, low dimension) seen in other works. 
  \item {\em Theoretical.} 
  To present new theoretical results {\em in the limit of the kernel bandwidth $\varepsilon_n \rightarrow 0$} that shed light on why RKHS/SVM methods often are so successful for complex data with sparse structure (Theorem~\ref{theorem::loss_epskernel} and Corollary \ref{corollary:manifold_example}), and to link smoothness of the regression with respect to $P$ to the approximation error of spectral series (Theorem~\ref{prop::errordecay}).
 \item {\em Experimental.} To systematically compare the statistical as well as the computational performance of spectral series and other methods using simulated and real-world data. Competing estimators include classical kernel smoothing, k-nearest neighbors (kNN) regression, regularization in RKHS, and recent state-of-the-art manifold and local regression methods. We ask questions such as:  Do the methods scale well with increasing dimension $d$ and increasing sample size $n$? What is the estimated loss and what is the computational time?
 \end{enumerate}

\comment{
\begin{enumerate}
\item introduce a data-driven series method for nonparametric regression that can handle a range of high-dimensional data objects which are common in modern applications (e.g., images, spectra, trajectories and text data),  
\item analyze how the method depends on the tuning parameters and the choice of kernel,
\item show how the basis functions of the series approximation can adapt to sparse structure in the data, and explain when the regression function has a sparse representation in this basis.  
\end{enumerate}
}

The paper is organized as follows. In Sec.~\ref{sec::series_est}, we describe the construction of the spectral series method, including details on how to estimate relevant quantities from empirical data and how to tune model parameters. Sec.~\ref{sec:related_work} discusses the connection to related work in machine learning and statistics. In Sections~\ref{sec::kernel} and~\ref{sec::theory}, we discuss the choice of kernel, and provide theoretical guarantees on the spectral series method. 
Finally, in Sec.~\ref{sec:examples}, we compare the performance of spectral series and other nonparametric estimators for a wide variety of data sets. 
 
%
 




\section{Orthogonal Series Regression} \label{sec::series_est}

\subsection{General Formulation} \label{sec::setup}
In low dimensions, orthogonal series has proved to be a powerful technique for nonparametric curve estimation~\cite{efromovich}. In higher dimensions, there is the question of whether one can find an appropriate basis and actually construct a series estimator that performs well. The general set-up of an orthogonal series regression is otherwise simple: Let $X_1, \ldots, X_n$ be an iid sample from a distribution $P$ with compact support  $\X \subset \R^d$.
 Suppose we have a real-valued response $$Y_i=f(X_i)+ \epsilon_i,$$ 
 where $f $ is an unknown function, and $\epsilon_i$ denotes iid random noise with mean zero and variance $\sigma^2$.
Our goal is to estimate the regression function $f(x)=\mathbb{E}(Y|X=x)$ in situations where $d$ is large and the data have sparse (i.e., low-dimensional)   structure.

Let $\{ \psi_{j}\}_{j \in \mathbb{N}}$ be an orthonormal basis of some appropriate Hilbert space $\mathcal H$ 
with inner product $\langle \cdot,\cdot \rangle_\mathcal{H}$ and norm $\| \cdot \|_\mathcal{H}$.  
We consider estimators of the form 
\begin{equation}
\hat{f}(x) =  \sum_{j=0}^{J} \hat \beta_{j} \hat \psi_{j}(x) ,  \label{eq::series_est}
\end{equation}
where $J$ is a smoothing parameter, and $\hat \psi_{j}$ and $\hat \beta_j$, in the general case, are data-based estimators of the basis functions $\psi_{j}$ and the expansion coefficients $\beta_j=\langle f, \psi_{j} \rangle_{\mathcal{H}}$.

\subsection{What Basis?}
A challenging problem is how to choose a good basis. The standard approach in nonparametric curve estimation is to choose a fixed known basis  $\{ \psi_{j}\}_{j \in \mathbb{N}}$ for, say, $L^2([0,1])$, such as a Fourier or wavelet basis. There is then no need to estimate basis functions. 
  In theory, such an approach can be extended to, eg., $L^2([0,1]^d)$ by a tensor product,\footnote{Traditional orthogonal series estimators require $d-1$ tensor products in  $\mathbb{R}^d$. For instance, if $d=2$, then it is common to choose a basis of the form
 $$\left\{\psi_{i,j}(\x)=\psi_{i}(x_1)\psi_{j}(x_2): i,j\in \mathbb{N}\right\},$$
where $\x=(x_1,x_2)$, and $\left\{\psi_{i}(x_1)\right\}_{i}$ and $\left\{\psi_{j}(x_2)\right\}_{j}$ are bases for functions in $L^2(\mathbb{R})$.} but tensor-product bases, as well as more sophisticated adaptive grid or triangulation methods (see~\cite{Mallat:2009} and references within), quickly become unusable for even as few as $d=5$ dimensions. 
  
What basis should one then choose when the dimension $d$ is large,  say, $d \sim 1000$? Ideally, the basis should be able to adapt to the underlying structure of the data distribution. This means: The basis should be orthogonal with respective to the distribution $P$ that generates the data, as opposed to the standard Lebesgue measure of the ambient space; the basis vectors should be concentrated around high-density regions where most of the ``action'' takes place; and the performance of the final series estimator should depend on the intrinsic rather than the ambient dimension of the data. In what follows, we present a spectral series approach where the unknown function is expanded into the estimated eigenfunctions of a kernel-based integral operator. As we shall see, the proposed estimator has many of the properties listed above.


\subsection{Construction of Adaptive Basis}~\label{sec::construction_basis}
Our starting point is a symmetric and positive semi-definite (psd) so-called Mercer kernel $k: \X \times \X \rightarrow \mathbb R$. These kernels include covariance functions and polynomial kernels, but we are in this work primarily interested in {\em local}, radially symmetric kernels $k_{\varepsilon}(x,y) = g\left( \frac{\| x-y\|}{\sqrt{\varepsilon}} \right)$,\footnote{Depending on the application, one can replace the Euclidean distance $\| x-y\|$ with a dissimilarity measure $d(x,y)$ that better reflects the distance between two data objects $x$ and $y$.}
 where $\varepsilon$ is a parameter that defines the scale of the analysis, and the elements $k(x,y)$ are positive and bounded for all  $x,y \in \X$.
 To simplify the theory, we renormalize the kernel according to
\begin{equation}
a_{\varepsilon}(x,y) = \frac{k_{\varepsilon}(x,y)}{p_\varepsilon(x)}, \label{eq::normalizedkernel}
 \end{equation}
where  $p_\varepsilon(x) = \int_{\X} k_\varepsilon(x,y)dP(y)$. This normalization is common in spectral clustering because it yields eigenvectors that act as indicator functions of connected components~\cite[Section 3.2]{Luxburg:2007}. The same normalization is also implicit in traditional Nadaraya-Watson kernel smoothers, which compute the local average at a point $x$ by weighting surrounding points $x_i$ by $\frac{k_{\varepsilon}(x,x_i)}{\sum_i k_{\varepsilon}(x,x_i)}$. 

We will refer to $a_{\varepsilon}(x,y)$ in Eq.~\ref{eq::normalizedkernel} as the {\em diffusion kernel}. The term ``diffusion'' stems from a random walks view over the sample space~\citep{MeilaShi:2001,LeeWasserman:2010}: One imagines a Markov chain on $\mathcal{X}$ with transition kernel $\Omega_\varepsilon(x,A) = \mathbb{P}(x \to A) =\int_A a_{\varepsilon}(x,y) dP(y)$. Starting at x, this chain moves to points $y$ close to $x$, giving preference to points with high density $p(y)$. The chain essentially encodes the ``connectivity'' of the sample space relative to $p$, and it has a unique stationary distribution $S_\varepsilon$ 
given by
$$
S_\varepsilon (A) = \frac{\int_A p_\varepsilon(x) dP(x)}{\int_{\X} p_\varepsilon(x) dP(x)},$$ 
where 
$S_\varepsilon(A) \to \frac{\int_A p(x) dP(x)}{\int_{\X} p(x) dP(x)}$ as $
\varepsilon \to 0.$ For finite $\varepsilon$, the stationary distribution $S_\varepsilon$ is a {\em smoothed} version of $P$.

 In our regression setting, we seek solutions from a Hilbert space 
  associated with  the kernel $a_{\varepsilon}$.
\comment{As $\widetilde{k}_{\varepsilon}$  is positive semi-definite,  there exists a unique reproducing kernel Hilbert space (RKHS) ${\mathcal{H}_k} \subset L^2(\X)$ for which  $\widetilde{k}_{\varepsilon}$  is the reproducing kernel~\citep{}.  We will look for solutions to the regression problem from a finite-dimensional subspace of this RKHS. The details are as follows: }
Following~\cite{LeeWasserman:2010}, we define a ``diffusion operator'' $A_{\varepsilon}$ --- which maps a function $f$ to a new function $A_{\varepsilon}  f$ --- according to 
\begin{equation}
 A_{\varepsilon}  f(x) = \int_{\mathcal{X}} a_\varepsilon(x,y) f(y) dP(y), \  \text{for} \ x \in \X.  \label{eq::Aeps}
 \end{equation}
The operator $A_{\varepsilon}$ has a discrete set of non-negative eigenvalues 
$\lambda_{\varepsilon,0}=1 \geq \lambda_{\varepsilon,1} \geq \ldots \geq 0$ with associated eigenfunctions $\psi_{\varepsilon,0}, \psi_{\varepsilon,1}, \ldots$, which we for convenience normalize to have unit norm. These eigenfunctions have two very useful properties: First, they are orthogonal with respect to the  density-weighted $L^2$ inner product 
$$\langle f, g\rangle_{\varepsilon} = \int_\X f(x) g(x) dS_\varepsilon(x);$$
that is, 
$$\langle \psi_{\varepsilon,i}, \psi_{\varepsilon,j} \rangle_{\varepsilon} = \delta_{i,j}.$$
Second, they also form a set of oscillatory functions which are concentrated around high-density regions. By construction, $\psi_{\varepsilon,0}$ is a constant function,   and the higher-order eigenfunctions are increasingly oscillatory. Generally speaking, $\psi_{\varepsilon,j}$ is the smoothest function relative to $P$, subject to being orthogonal to $\psi_{\varepsilon,i}$ for $i<j$.

{\bf Interpretation.} 
 The diffusion operator $A_{\varepsilon}$ and its eigenfunctions contain information about the connectivity structure of the sample space. 
 There are two ways one can view the  eigenfunctions  $\psi_{\varepsilon,0}, \psi_{\varepsilon,1},  \psi_{\varepsilon,2}, \ldots$: 
\begin{enumerate}[(i)]
\item The eigenfunctions define {\em new coordinates of the data}. If the data $x$ represent high-dimensional complex objects, there is often no simple way of ordering  the data. However, by a so-called ``eigenmap''
\begin{equation}
 	x \mapsto (\psi_{\varepsilon,1}(x), \psi_{\varepsilon,2}(x), \ldots, \psi_{\varepsilon,J}(x)),  \label{eq:eigenmap}
\end{equation}  
one can transform the data into an embedded space where points that are highly connected are mapped close to each other \cite{LafonLee2006}.
The eigenmap can be used for data visualization as in Fig.~\ref{fig::diffMapISO} and Fig.~\ref{fig::diffusionMapSpec}. If we choose $J<d$, then we are effectively reducing the dimensionality of the problem by mapping the data from $\mathbb{R}^d$ to $\mathbb{R}^J$.

\item The eigenfunctions form {\em a Hilbert basis for functions on the data}. More specifically, the set $\psi_{\varepsilon,0}, \psi_{\varepsilon,1}, \ldots$ is an orthogonal basis of 
$L^2(\X, P)$.
The value of this result is that we can express most physical quantities that vary as a function of the data as a series expansion of the form 
 $f(x) =  \sum_{j=0}^{\infty}  \beta_{\varepsilon,j} \psi_{\varepsilon,j}(x)$. 
\end{enumerate}
 In this work, we study the second point of view and its implications on nonparametric estimation in high dimensions. 

\subsection{Estimating the Regression Function from Data}~\label{sec::regression_estimate} 
In practice, of course, we need to estimate the basis $\{\psi_{\varepsilon,j}\}_{j}$ and the projections $\{\beta_{\varepsilon,j}\}_j$ from data. In this section, we describe the details. 

Given $X_1,\ldots,X_n$, we compute 
 a row-stochastic matrix  $\mathbb{A}_\varepsilon$, 
 where  
\begin{equation}
\mathbb{A}_\varepsilon(i,j)= \frac{k_\varepsilon(X_i,X_j)}{\sum_{l=1}^{n} k_\varepsilon(X_i,X_l)} \label{eq::A_matrix}
\end{equation}
for $i,j=1,\ldots,n$. 
 The elements $\mathbb{A}_\varepsilon(i,j)$ can be interpreted as transition probabilities $\mathbb{A}_\varepsilon(i,j) =  \mathbb{P}(x_i \rightarrow x_j)$ for a Markov chain over the data points (i.e., this is the discrete analogue of Eq.~\ref{eq::normalizedkernel} and a diffusion over $\mathcal{X}$). Let $\widehat{p}_\varepsilon(x)= \frac{1}{n}\sum_{j=1}^{n} k_\varepsilon(x,X_j)$. The Markov chain has 
a unique stationary measure given by $(\hat s_\varepsilon(X_1),\ldots,\hat s_\varepsilon(X_n))$, where the $i$th element
\begin{equation}\label{eq::s_hat}
\hat s_\varepsilon(X_i)=\frac{\widehat{p}_\varepsilon(X_i)}{\sum_{j=1}^{n} \widehat{p}_\varepsilon(X_j)}
\end{equation} is a kernel-smoothed density estimate at the $i$th observation.

To estimate the eigenfunctions $\psi_{\varepsilon,1},\ldots,\psi_{\varepsilon,J}$ of the continuous diffusion operator $A_\varepsilon$ in Eq.~\ref{eq::Aeps}, we first calculate the  eigenvalues $\lambda^{\mathbb{A}}_{\varepsilon,1},\ldots,\lambda^{\mathbb{A}}_{\varepsilon,J}$ and the associated (orthogonal) eigenvectors $\tilde{\psi}_{\varepsilon,1}^{\mathbb{A}},\ldots,\tilde{\psi}_{\varepsilon,J}^{\mathbb{A}}$ 
 of the symmetrized kernel matrix $\tilde{\mathbb{A}}_\varepsilon$, where
\begin{equation}\label{eq::symmetrized_matrix}
\tilde{\mathbb{A}}_\varepsilon(i,j) = \frac{k_\varepsilon(X_i,X_j)}{\sqrt{\sum\limits_lk_\varepsilon(X_i,X_l)} \sqrt{\sum\limits_lk_\varepsilon(X_l,X_j)}}.
\end{equation}
We normalize the eigenvectors so that
$\displaystyle\frac{1}{n}\sum\limits_{i=1}^n \tilde{\psi}_{\varepsilon,j}^{\mathbb{A}}(i) \tilde{\psi}_{\varepsilon,k}^{\mathbb{A}}(i)=\delta_{j,k}$, and define the new vectors 
$\displaystyle{\psi}_{\varepsilon,j}^{\mathbb{A}}(i)=\frac{\tilde{\psi}_{\varepsilon,j}^{\mathbb{A}}(i)}{\sqrt{\hat s_\varepsilon(X_i)}}$ for $i=1,\ldots,n$ and $j=1,\ldots,J$.  By construction, 
it holds that the $\lambda_{\varepsilon,j}^{\mathbb{A}}$'s and $\psi_{\varepsilon,j}^{\mathbb{A}}$'s are
 eigenvalues and right eigenvectors of the Markov matrix $\mathbb{A}_\varepsilon$: 
\begin{equation}\label{eq::Aeps_psi}
\mathbb{A}_\varepsilon {\psi}_{\varepsilon,j}^{\mathbb{A}} = \lambda_{\varepsilon,j}^{\mathbb{A}} {\psi}_{\varepsilon,j}^{\mathbb{A}}
\end{equation}
where 
\begin{equation}
\displaystyle\frac{1}{n}\sum\limits_{i=1}^n \psi_{\varepsilon,j}^{\mathbb{A}}(i) \psi_{\varepsilon,k}^{\mathbb{A}}(i)\hat s_\varepsilon(X_i) = \delta_{j,k}. \label{eq::eig_normalization}
\end{equation}
Note that the $n$-dimensional vector $\psi_{\varepsilon,j}^{\mathbb{A}}$  can be regarded as estimates 
of $\psi_{\varepsilon,j}(x)$ at the observed values $X_1,\ldots, X_n$. In other words, let 
\begin{equation}\label{eq::eigest}\hat\lambda_{\varepsilon,j}\equiv \lambda_{\varepsilon,j}^{\mathbb{A}} \  \ \ {\rm and} \ \ \ \hat\psi_{\varepsilon,j}(X_i) \equiv \psi_{\varepsilon,j}^{\mathbb{A}}(i)\end{equation} for $i=1,\ldots,n$.
We estimate the  function $\psi_{\varepsilon,j}(x)$ at values of $x$ not corresponding to one of the $X_i$'s using the so-called Nystr\"{o}m method. The idea is to
first rearrange the eigenfunction-eigenvalue equation
$\lambda_{\varepsilon,j} \psi_{\varepsilon,j} = A_\varepsilon \psi_{\varepsilon,j}$ 
as  
$$\psi_{\varepsilon,j}(x) = \frac{A_\varepsilon \psi_{\varepsilon,j}}{\lambda_{\varepsilon,j}}  =\frac{1}{\lambda_{\varepsilon,j}}  \int_{\X} \frac{k_\varepsilon(x,y)}{\int_{\X} k_\varepsilon(x,y)dP(y)} \psi_{\varepsilon,j}(y) dP(y),$$
and use the kernel-smoothed estimate
\begin{equation}\label{eq::nystrom}
\hat\psi_{\varepsilon,j}(x)  =\frac{1}{ \hat{\lambda}_{\varepsilon,j}} \sum_{i=1}^{n}  \frac{k_\varepsilon(x,X_i)}{\sum_{l=1}^{n} k_\varepsilon(x,X_l)} \hat\psi_{\varepsilon,j}(X_i). 
\end{equation}
for $\hat{\lambda}_{\varepsilon,j}>0$.

 Our final regression estimator is defined by Eq.~\ref{eq::series_est}
 with the estimated eigenvectors in Eq.~\ref{eq::nystrom} and expansion coefficients computed according to
 \begin{equation} 
 \hat{\beta}_{\varepsilon,j}  =  \frac{1}{n} \sum_{i=1}^n Y_i \hat \psi_{\varepsilon,j}(X_i) \hat{s}_\varepsilon(X_i). \label{eq::betahat}
\end{equation}

\comment{The estimated eigenfunctions $\hat \psi_j(x)$ should be approximately orthogonal on the labeled data set.
Hence, semi-supervised learning (SSL), where one computes and normalizes the eigenfunctions using both labeled and unlabeled data  (Eq.\ref{eq::eig_normalization}), can lead to biased estimates of the coefficients and poor regression results. In those situations, it may be more advantageous to either use labeled data only, or to compute a weighted least squares regression or equivalent without the orthogonality condition (ADD equation references).}


\begin{remark}[{\bf Semi-Supervised Learning, SSL}]\label{remark::SSL}
The spectral series framework naturally extends to  {\em semi-supervised learning} (SSL)
  ~\citep{zhu2009introduction} where in addition to the ``labeled'' sample $(X_1,Y_1),\ldots,(X_n,Y_n)$ there are
  additional ``unlabeled'' data; i.e., data  $X_{n+1},\ldots,X_{n+m}$ where the covariates $X_i$ but not the labels $Y_i$ are known. Typically $m \gg n$, as collecting data often is less costly than labeling them.   
   By including unlabeled examples (drawn from the same distribution $P_X$) into the kernel matrix $\mathbb{A}_\varepsilon$, we can 
   improve our estimates of $\lambda_{\varepsilon,j}$, $\psi_{\varepsilon,j}$ and $S_\varepsilon$. 
    The summation in Equations~\ref{eq::eig_normalization} and~\ref {eq::nystrom} will then be over all $n+m$ observations, while Eq.~\ref{eq::betahat} remains the same as before. See e.g.~\cite{Nadler_semi, Zhou:Srebro:2011} for SSL with Laplacian eigenmaps in the limit  of infinite unlabeled data, i.e., in the limit $m \rightarrow \infty$.
\end{remark}

\subsection{Loss Function and Tuning of Parameters}~\label{sec::loss_function}
We measure the performance of an estimator $\hat{f}(x)$ via the $L^2$ loss function
 \begin{equation}
L(f,\hat f) = \int_{\X} 
 \left( f(x) - \hat{f}(x)\right)^2 dP(x). \label{eq:lossfunction}
\end{equation}
  To choose tuning parameters (such as the kernel bandwidth $\varepsilon$ and the number of basis functions $J$), we split the data into a training and a validation set.  For each choice of $\varepsilon$ and a sufficiently large constant $J_{\rm max}$, we use the training set and Eqs.~\ref{eq::nystrom}-\ref{eq::betahat} to estimate the eigenvectors $\psi_{\varepsilon,1}, \ldots, \psi_{\varepsilon,J_{\rm max}}$ and the expansion coefficients $\beta_{\varepsilon,0},\ldots, \beta_{\varepsilon,J_{\rm max}}$. We then use the validation set $(X_1',Y_1'),\ldots,(X_n',Y_n')$ to compute the estimated loss
\begin{equation}\label{eq::empirical_loss}
 \hat{L}(f,\hat f) = \frac{1}{n} \sum_{i=1}^n \left( Y_i' - \hat{f}(X_i')\right)^2 =   \frac{1}{n} \sum_{i=1}^n \left( Y_i' - \sum_{j=0}^{J} \hat \beta_{\varepsilon,j} \hat \psi_{\varepsilon,j}(X_i') \right)^2
 \end{equation}
for different values of $J \leq J_{\rm max}$. We choose the  ($\varepsilon$, $J$)-model with the lowest estimated loss on the validation set.
 
  The computation for fixed $\varepsilon$ and different $J$ is very fast. Due to orthogonality of the basis, the estimates $\hat \beta_{\varepsilon,j}$ and  
 $\hat \psi_{\varepsilon,j}$ depend on $\varepsilon$ but {\em not} on $J$.
 

\subsection{Scalability}~\label{sec::scalability} 
The spectral series estimator is faster than most traditional approaches in high dimensions. Once the kernel matrix has been constructed, the eigendecomposition takes the {\em same} amount of time for all values of $d$.

 In terms of scalability for large data sets,  
  one can dramatically reduce the computational cost  by implementing fast approximate eigendecompositions. 
 For example, the Randomized SVD by Halko et al.~\citep{Halko} cuts down the cost from $O(n^3)$ to roughly $O(n^2)$ with little impact on statistical performance (see Fig.~\ref{fig::RSVD}).  
 According to Halko et al., these randomized methods are especially well-suited for parallel implementation, which is a topic we will explore in future work.   

\section{Connection to Other Work}\label{sec:related_work} 

\subsection{Linear Regression with Transformed Data}
 One can view our series model  as a (weighted) linear regression 
  after a data transformation $Z=\Psi(X)$, where $\Psi=(\psi_{1},\ldots,\psi_{J})$ are the first $J$ eigenvectors of the diffusion operator $A_\varepsilon$. 
By increasing $J$, the dimension of the feature space, we achieve more flexible, {\em fully  nonparametric} representations. Decreasing $J$ adds more structure to the regression, as dictated by the eigenstructure of the data.

Eq.~\ref{eq::betahat} is similar to a weighted least squares (WLS) solution to a linear regression in $(Z_1, Y_1),\ldots,(Z_n, Y_n)$ but with an efficient orthogonal series implementation and no issues with collinear variables. Define the $n \times (J+1)$ matrix of predictors,
\begin{equation}
 \mathbb{Z} =
 \begin{pmatrix}
  1 &  {\psi}_1(X_1) & \cdots & {\psi}_{J}(X_1) \\
  1 &  {\psi}_1(X_2) & \cdots & {\psi}_{J}(X_2) \\
  \vdots  & \vdots  & \ddots & \vdots  \\
 1 & {\psi}_1(X_n) & \cdots & {\psi}_{J}(X_n) 
 \end{pmatrix},
\end{equation}
and introduce the weight matrix
\begin{equation}
 \mathbb{W} =
 \begin{pmatrix}
{s}(X_1)   &  0 & \cdots & 0 \\
 0 &  {s}(X_2)  & \cdots & 0\\
  \vdots  & \vdots  & \ddots & \vdots  \\
 0 &  0 & \cdots &  {s}(X_n)
 \end{pmatrix}, \label{eq:WLS_weightmatrix}
\end{equation}
where $\Psi_j$ and $s$ are estimated from data (Equations~\ref{eq::s_hat} and~\ref{eq::eigest}).  Suppose that $Y=\mathbb{Z} \beta+e$, 
where $Y=(Y_1,\ldots,Y_n)^T$,  $\beta=(\beta_1,\ldots,\beta_J)^T$, and the random vector $e=(\epsilon_1,\ldots,\epsilon_n)^T$ represents the errors.  By minimizing the weighted residual sum of squares \begin{equation}RSS(\beta) = (Y-\mathbb{Z} \beta)^T \mathbb{W} (Y-\mathbb{Z} \beta), \end{equation}
we arrive at the WLS estimator
\begin{equation}
 \hat \beta = (\mathbb{Z}^T \mathbb{W} \mathbb{Z})^{-1} (\mathbb{Z}^T \mathbb{W} Y) = \frac{1}{n} \mathbb{Z}^T \mathbb{W} Y ,  \label{eq:WLS_coeffs}
\end{equation}
where the matrix $\mathbb{W}$ puts more weight on observations in high-density regions.
This expression is equivalent to Eq.~\ref{eq::betahat}.  

 Note that thanks to the orthogonality property $\mathbb{Z}^T \mathbb{W} \mathbb{Z}=n\mathbb{I}$, model search and model selection are feasible even for complex models with very large $J$. This is in clear contrast with standard multiple regression where one needs to recompute the $\hat \beta_j$ estimates for each model with a different $J$, invert the matrix $\mathbb{Z}^T \mathbb{W} \mathbb{Z}$, and potentially deal with inputs (columns of the design matrix $\mathbb{Z}$) that are linearly dependent.

\begin{remark} [{\bf Heteroscedasticity}]
More generally, let $\sigma(x)$ be a non-negative function rather than a constant, and let $\epsilon_i$ be iid realizations of a random variable $\epsilon$ with zero mean and unit variance. Consider the regression model $Y_i=f(X_i)+\sigma(X_i)\epsilon_i$. We can handle heteroscedastic errors by applying the same framework as above to a {\em rescaled} regression function $g(x)=f(x)/\sigma(x)$.
\end{remark}

\subsection{Kernel Machine Learning and Regularization in RKHS} \label{sec:kernel_ML}

 Kernel-based regularization methods use similar approximation spaces as us. 
In kernel machine learning~\citep{Scholkopf:Smola:2001,cucker2007learning}, one often considers the variational problem
\begin{equation} 
\min_{f \in {\mathcal{H}_k}} \left[ \frac{1}{n}\sum_{i=1}^n L(y_i,f(x_i))+ \gamma \|f\|_{\mathcal{H}_k}^2\right]  ,\label{eq::regularization} 
\end{equation}
where $L(y_i,f(x_i))$ is a convex loss function, $\gamma>0$ is a penalty parameter, and 
$\mathcal{H}_k$ is the Reproducing Kernel Hilbert Space (RKHS) associated with a symmetric positive semi-definite kernel $k$.\footnote{To every continuous, symmetric, and positive semi-definite kernel $k:\mathcal{X} \times \mathcal{X} \rightarrow \mathbb R$ is associated a unique RKHS $\mathcal{H}_k$~\citep{Aronszajn:1950}. This RKHS is defined to be the closure of the linear span of the set of functions $\{k(x, \cdot) : x \in \X \}$ with the inner product satisfying the reproducing property $\langle k(x,\cdot), f \rangle_{\mathcal{H}_k}=f(x)$ for all $x \in \X, f \in \mathcal{H}_k$.} Penalizing the RKHS norm  $\| \cdot\|_{\mathcal{H}_k}$ imposes smoothness conditions on possible solutions. 
Now suppose that
$$k(x,y) = \sum_{j=0}^{\infty} \lambda_{j} \phi_j(x) \phi_j(y),$$
where the RKHS inner product is related to the $L^2$-inner product according to $ \ \langle \phi_i, \phi_{j} \rangle_{\mathcal{H}_k} = \frac{1}{\lambda_i}\langle \phi_i, \phi_j \rangle_{L^2(\X, P)}  =  \frac{1}{\lambda_i}\delta_{i,j}.$  Eq.~\ref{eq::regularization} is then equivalent to considering eigen-expansions
$$ f(x) = \sum_{j=0}^{\infty} \beta_j \phi_{j}(x),$$
and seeking solutions to
$ \min_{f \in {\mathcal{B}_r}} \frac{1}{n }\sum_{i=1}^n L(y_i,f(x_i)),$
where the hypothesis space
\begin{equation} B_r=\left\{f \in \mathcal{H}_k : \|f\|_{\mathcal{H}_k} \leq r\right\} \label{eq::RKHS_ball} \end{equation}
is a ball of the RKHS $\mathcal{H}_k$ with radius $r$, and the RKHS norm is given by
  $\|f\|_{\mathcal{H}_k} = \left(\sum_{j=0}^{\infty} \frac{\beta_j^2}{\lambda_{j}}\right)^{1/2}$. 


Here are some key observations: 

(i) The above setting is similar to ours. The regularization in Eq.~\ref{eq::regularization} differentially shrinks contributions from higher-order terms with small $\lambda_j$ values. In spectral series, we use a projection (i.e., a basis subset selection) method, but the empirical performance is usually similar. 

(ii) There are some algorithmic differences, as well as differences in how the two regression estimators are analyzed and interpreted. In our theoretical work,  
we consider Gaussian kernels with {\em flexible variances}; that is, we choose the approximation spaces in a data-dependent way (cf. multi-kernel regularization schemes for SVMs~\cite{Wu:EtAl:2007}) so that the estimator can adapt to sparse structure and the intrinsic dimension of the data. Most theoretical work in kernel machine learning assume a fixed RKHS. 

(iii) There are also other differences. 
Support Vector Machines~\citep{Steinwart:Christmann:2008} and other kernel-based regularization methods (such as splines, ridge regression and radial basis functions) never explicitly compute the eigenvectors of the kernel. Instead, these methods rely on the classical Representer Theorem~\citep{Wahba:1990} which states that the solution to Eq.~\ref{eq::regularization} is a finite expansion of the form $f(x)=  \sum_{i=1}^n \alpha_i k(x_i,x)$. 
The original infinite-dimensional variational problem is then reduced to a finite-dimensional optimization of the coefficients $\alpha_i$. In a naive least-squares implementation, 
however, one has to recompute these coefficients for each choice of the penalty parameter $\gamma$, which can make cross-validation cumbersome. In our spectral series approach, we take advantage of the orthogonality of the basis for fast model selection and computation of the  $\beta_j$ parameters. As in spectral clustering, we also use eigenvectors to {\em organize} and {\em visualize} data that can otherwise be hard to interpret.



\subsection{Manifold Regularization and Semi-Supervised Learning}
 Our spectral series method is closely related to Laplacian-based regularization:
 In~\cite{Belkin:EtAl:06}, Belkin et al. extend the kernel-based regularization framework to
 incorporate additional information about the geometric structure of the marginal $P_X$.  Their idea is to add a data-dependent penalty term to  Eq.~\ref{eq::regularization} that controls the complexity as measured by the geometry of the distribution. Suppose that one is given labeled data $(X_1,Y_1), \ldots, (X_n,Y_n) \sim P_{X,Y}$ as well as unlabeled data $X_{n+1}, \ldots, X_{n+m} \sim P_X$, where in general $m \gg n$. (The limit $m \rightarrow \infty$ corresponds to having full knowledge of $P_X$.) Under the assumption that the support of $P_X$ is a compact submanifold of $\mathbb{R}^d$,  the authors propose minimizing a graph-Laplacian regularized least squares function
\begin{equation} 
	\min_{f \in {\mathcal{H}_k}} \left[ \frac{1}{n}\sum_{i=1}^n L(y_i,f(x_i))+ \gamma_{A} \|f\|_{\mathcal{H}_k}^2 + 
	 \frac{\gamma_{I}}{n+m} \sum_{i,j=1}^{n+m} (f(x_i)-f(x_j))^2 W_{i,j}  \right]  ,\label{eq::manifold_reg} 
\end{equation}
where  $W_{i,j}$ are the edge weights in the graph, and the last Laplacian penalty term favors functions $f$ for which $f(x_i)$ is close to $f(x_j)$ when $x_i$ and $x_j$ are connected with large weights. 

\comment{{\color{red}[[ATTN: SKIP THIS?} However, if the same kernel $k$ and bandwidth $\varepsilon$ is used to define the weights $W_{i,j}$ as for the RKHS norm, then the Laplacian regularized optimization problem reduces to the standard RKHS regularization problem in Eq.~\ref{eq::regularization} (see Remark 4 in Belkin et al). This behavior is consistent with the result in [cite Lafferty and Wasserman, 2007] that a trivially defined estimator that uses Laplacian regularization has the same rate of convergence as the usual kernel estimator. In other words, the interesting scenario (where unlabeled data could potentially improve the estimator) is when the Laplacian penalty term is based on a sharper kernel than in the RKHS norm, with a bandwidth  that approaches $0$ in the limit  of infinite unlabeled data.{\color{red}]]} }

 Note that the eigenbasis of our row-stochastic matrix $\mathbb{A}_\varepsilon$ minimizes the distortion $\sum_{i,j} (f(x_i)-f(x_j))^2 W_{i,j}$ in Eq.~\ref{eq::manifold_reg} if you regard the entries of $\mathbb{A}_\varepsilon$ as the weights  $W_{i,j}$~\citep{BelkinNiyogi03}. Indeed, the eigenvector $\psi_{\varepsilon,j}$ minimizes the term subject to being orthogonal to $\psi_{\varepsilon,i}$ for $i<j$.   Hence, including a Laplacian penalty term is comparable to truncating the eigenbasis expansion in spectral series.  Moreover, the semi-supervised regularization in Eq.~\ref{eq::manifold_reg} is similar to a semi-supervised version of our spectral series approach, where we first use both labeled and unlabeled data and a kernel with bandwidth $\varepsilon_{m+n}$ to compute the eigenbasis, and then extend the eigenfunctions  according to Eq.~\ref{eq::nystrom} via a (potentially wider) kernel with bandwidth $h_n$. The main downside of the Laplacian-based framework above is that it is hard to analyze theoretically. As with other kernel-based regularizers, the method also does not explicitly exploit eigenvectors and orthogonal bases.




\section{Choice of Kernel}\label{sec::kernel}
In the RKHS literature, there is a long list of commonly used kernels. These include, e.g., the Gaussian kernel $k(x,y) =   \exp \left(-\frac{\|x-y\|^2}{4 \varepsilon} \right)$, polynomial kernels $k(x,y)=(\langle x,y\rangle + 1)^q$~\citep{Vapnik:1996}, and the thin-plate spline kernel $k(x,y)=\|x-y\|^2 \log(\|x-y\|^2)$~\citep{Girosi:95}.  In our numerical experiments (Sec.~\ref{sec:examples}), we will consider both Gaussian and polynomial kernels, but throughout the rest of the paper, we will primarily work with the (row-normalized) Gaussian kernel. 
There are several reasons for this choice: 
\begin{enumerate}[(i)] 
   \item The Gaussian kernel can be interpreted as the heat kernel in a manifold setting~\citep{BelkinNiyogi03,Grigoryan:06}. We will take advantage of this connection in the theoretical analysis of the spectral series estimator  (Sec.~\ref{sec::theory}). 
     \item 
     The eigenfunctions of the Gaussian kernel are simultaneously concentrated in time (i.e., space) and frequency, and are particularly well-suited for estimating functions that are smooth with respect to a low-dimensional data distribution.
\end{enumerate}

The following two examples illustrate some of the differences in the eigenbases of Gaussian and polynomial kernels:

\comment{The Gaussian kernel is {\em smooth} or infinitely differentiable. It can be interpreted as the {\em heat kernel} in a manifold setting~\citep{BelkinNiyogi03,Grigoryan:06}. The kernel is {\em universal}, meaning that the associated RKHS is dense in the space of continuous functions on $\mathcal{X}$~\citep{Micchelli:2006}. Finally, the computed eigenfunctions are simultanously {\em concentrated in time and frequency} --- a property which is especially useful in a setting with sparse structure in high dimensions. The following two examples illustrate 
how the Gaussian kernel yields smooth Fourier-like eigenfunctions concentrated around the underlying data distribution, whereas the eigenfunctions of polynomial kernels lack such a property. 
}

\begin{example}
Suppose that $P$ is a uniform distribution $U(-1,1)$  on the real line. Fig.~\ref{fig::unif_distribution}, {\em left}, shows the eigenfunctions of a third-order polynomial kernel $k(x,y)=(\langle x,y\rangle + 1)^3$. These functions are smooth but have large values outside the support  of $P$. Contrast this eigenbasis with the eigenfunctions in Fig.~\ref{fig::unif_distribution}, {\em right}, of a Gaussian kernel. The latter functions are concentrated on the support of $P$ and are orthogonal on $(-1,1)$ as well as on $(-\infty,\infty)$.
\begin{figure}[H]
\vspace{-0.7in}
\begin{center}
\begin{tabular}{cc}
{\hspace{-0.3in} \includegraphics[width=2.7in,angle=-0]{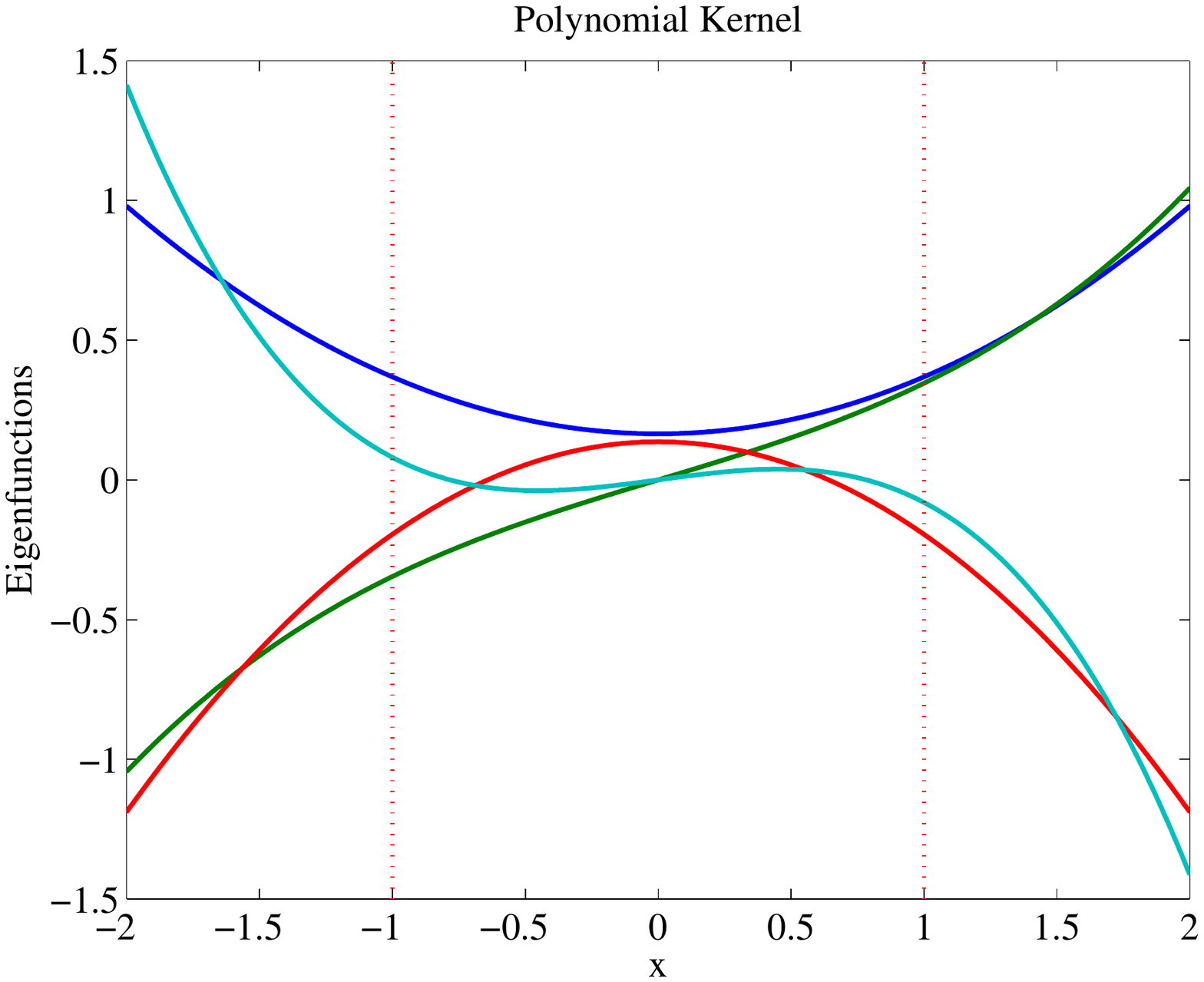}} &
{\hspace{-0.4in} \includegraphics[width=2.7in,angle=-0]{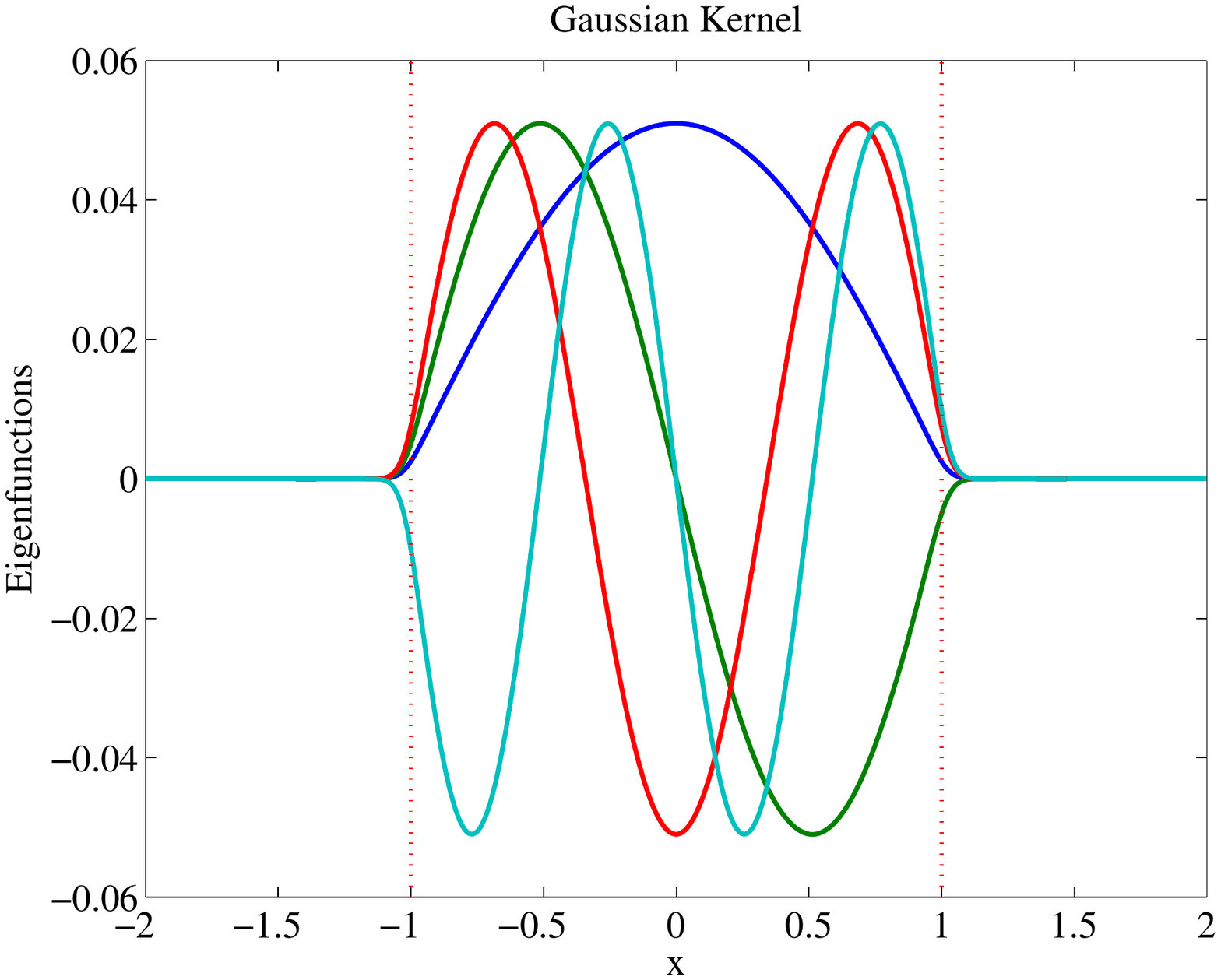}}
\end{tabular}
\end{center}
\vspace{-0.9in}
\caption{ \footnotesize Uniform distribution U(-1,1). Eigenfunctions of a third-order polynomial kernel, {\em left}, and of an (un-normalized) Gaussian kernel, {\em right}.  The latter eigenfunctions form an orthogonal Fourier-like basis concentrated on the support of the data distribution.}
\label{fig::unif_distribution}
\end{figure}
\end{example}

\begin{example}\label{eq::noisyspiral}
Consider data around a noisy spiral:
$$
\left\{
\begin{array}{ccc}
x(u) &=& \sqrt{u} \cos(\sqrt{u}) + \epsilon_x\\
y(u) &=& \sqrt{u} \sin(\sqrt{u})   + \epsilon_y,
\end{array}  
\right.
$$
where $u$ is a uniform random variable, and $\epsilon_x$ and  $\epsilon_y$ are normally distributed random variables. The eigenfunctions of a polynomial kernel do not adapt well to the underlying distribution of the data. Fig.~\ref{fig::spiral_data}, {\em left}, for example, is a contour plot of the Nystr\"{o}m extension of the fourth empirical eigenvector of a third-order polynomial kernel. In contrast, the eigenfunctions of a Gaussian diffusion kernel vary smoothly along the spiral direction,   forming a Fourier-like basis with orthogonal eigenfunctions that concentrate around high-density regions; 
 see Fig.~\ref{fig::spiral_data}, {\em right}.
\begin{figure}[H]
\vspace{-0.7in}
\begin{center}
\begin{tabular}{cc}
{\hspace{-0.5in} \includegraphics[width=3in,angle=-0]{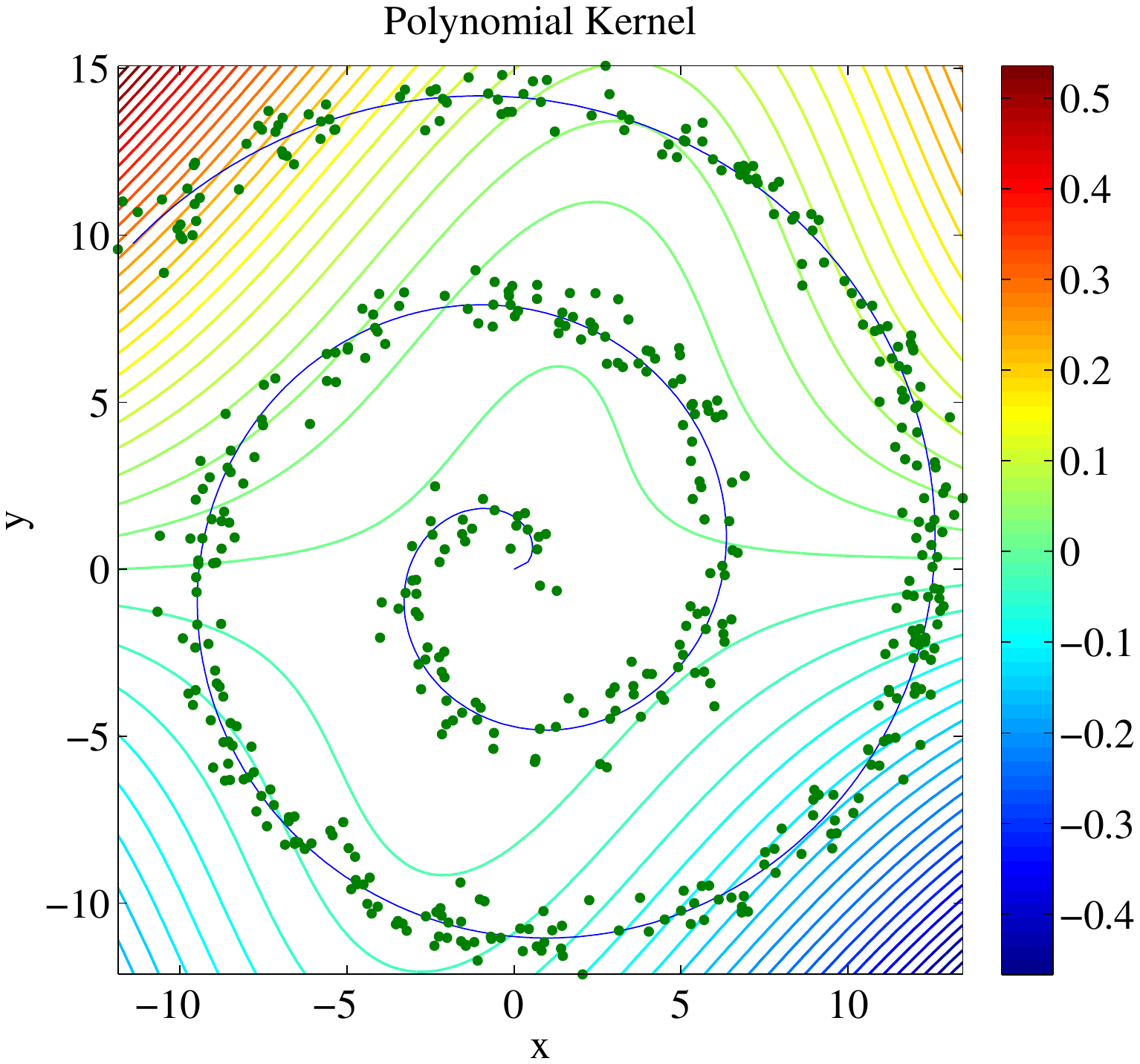}} &
{\hspace{-0.8in} \includegraphics[width=3in,angle=-0]{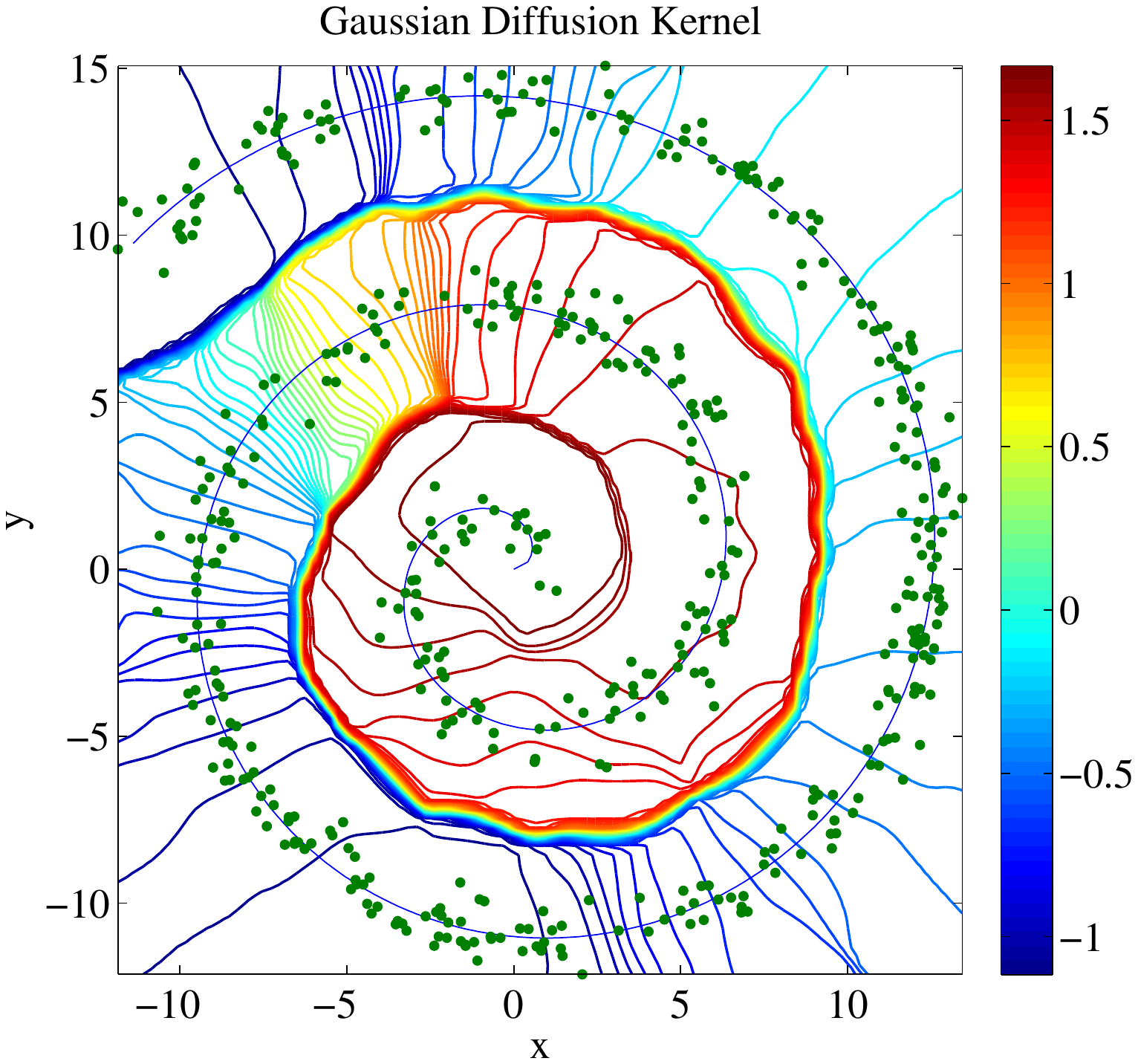}}
\end{tabular}
\end{center}
\vspace{-0.9in}
\caption{ \footnotesize Spiral data. Contour plots of the fourth eigenvector of a third-order polynomial kernel ({\em left}) and the fourth eigenvector of a Gaussian diffusion kernel ({\em right}). The latter eigenvector is localized and varies smoothly along the spiral direction.}
\label{fig::spiral_data}
\end{figure}
\end{example}

 In high dimensions, Gaussian extensions can be seen as a generalization of {\em prolate spheroidal wave functions} \cite{Coifman:Lafon:06b}. Prolates were originally introduced by Slepian and Pollack as the solution to the problem of simultaneously and optimally concentrating a function and its Fourier content (see~\cite{Slepian:1983} for a fascinating recount of this development in Fourier analysis and modeling).
    The band-limited functions that maximize their energy content within a space domain $\mathcal{X} \subset \mathbb{R}^n$ are extensions of the eigenfunctions of the integral operator of a Bessel kernel restricted to $\mathcal{X} $~\cite[Section 3.1]{Coifman:Lafon:06b}. In high dimensions, Bessel and Gaussian kernels are equivalent~\citep{Schoenberg:1938}, suggesting that the eigenfunctions of the Gaussian kernel are {\em  nearly}  optimal. 

However, although Gaussian kernels have many advantages, they may not always be the best choice in practice. Ultimately, this is determined by the application and by what the best measure of similarity between two data points would be.  Our framework suggests a principled way of selecting the best kernel for regression: Among a set of reasonable candidate kernels, choose the estimator with the {\em smallest empirical loss} according to Eq.~\ref{eq::empirical_loss}. We will, for example, use this approach in Sec.~\ref{sec:examples} to choose the optimal degree $q$ for a set of polynomial kernels of the form $k(x,y)=(\langle x,y\rangle + 1)^q$.

{\bf Normalization of Local Kernels.}  In the RKHS literature, it is standard to work with ``unnormalized'' kernels. In spectral clustering~\citep{Luxburg:2007}, on the other hand, researchers often use the ``stochastic'' and ``symmetric'' normalization schemes in Eq.~\ref{eq::A_matrix} and Eq.~\ref{eq::symmetrized_matrix}, respectively. We have found (Sec.~\ref{sec:examples}) that 
 the exact normalization often has little effect on the performance in regression. Nevertheless, we choose to use the {\em row-stochastic} kernel for reasons of interpretation and analysis: First, the limit of the bandwidth $\varepsilon \rightarrow 0$ is well-defined, and there is a series of works on the convergence of the graph Laplacian to the Laplace-Beltrami operator on Riemannian manifolds~\citep{Coifman:Lafon:06, Belkin:Niyogi:05, Hein:EtAl:05, Singer:06, Gine:Koltchinskii:2006}. Fourier functions originate from solving a Laplace eigenvalue problem on a bounded domain; hence, the eigenfunctions of the diffusion operator can be seen as a generalization of Fourier series to manifolds.  

Moreover, the row-stochastic kernel yields less variable empirical functions than the unnormalized or symmetric forms. As an illustration, consider the noisy spiral data  in Example~\ref{eq::noisyspiral}. Fig.~\ref{fig::spiral_projections} shows the estimated projections onto the spiral direction of the eigenfunctions of the {\em symmetric} and the  {\em stochastic} forms; see the left and right plots, respectively. The eigenfunctions  are clearly smoother in the latter case. By construction, the empirical eigenfunctions of the symmetric operator are orthogonal with respect to the empirical distribution $\hat{P}_n$, whereas the estimated eigenfunctions of the stochastic operator are orthogonal with respect to the {\em smoothed} data distribution $\hat{S}_{\varepsilon}$. The kernel bandwidth $\varepsilon_n$ defines the scale of the analysis.
\begin{figure}[H]
\vspace{-0.7in}
\begin{center}
\begin{tabular}{cc}
{\hspace{-0.3in} \includegraphics[width=2.7in,angle=-0]{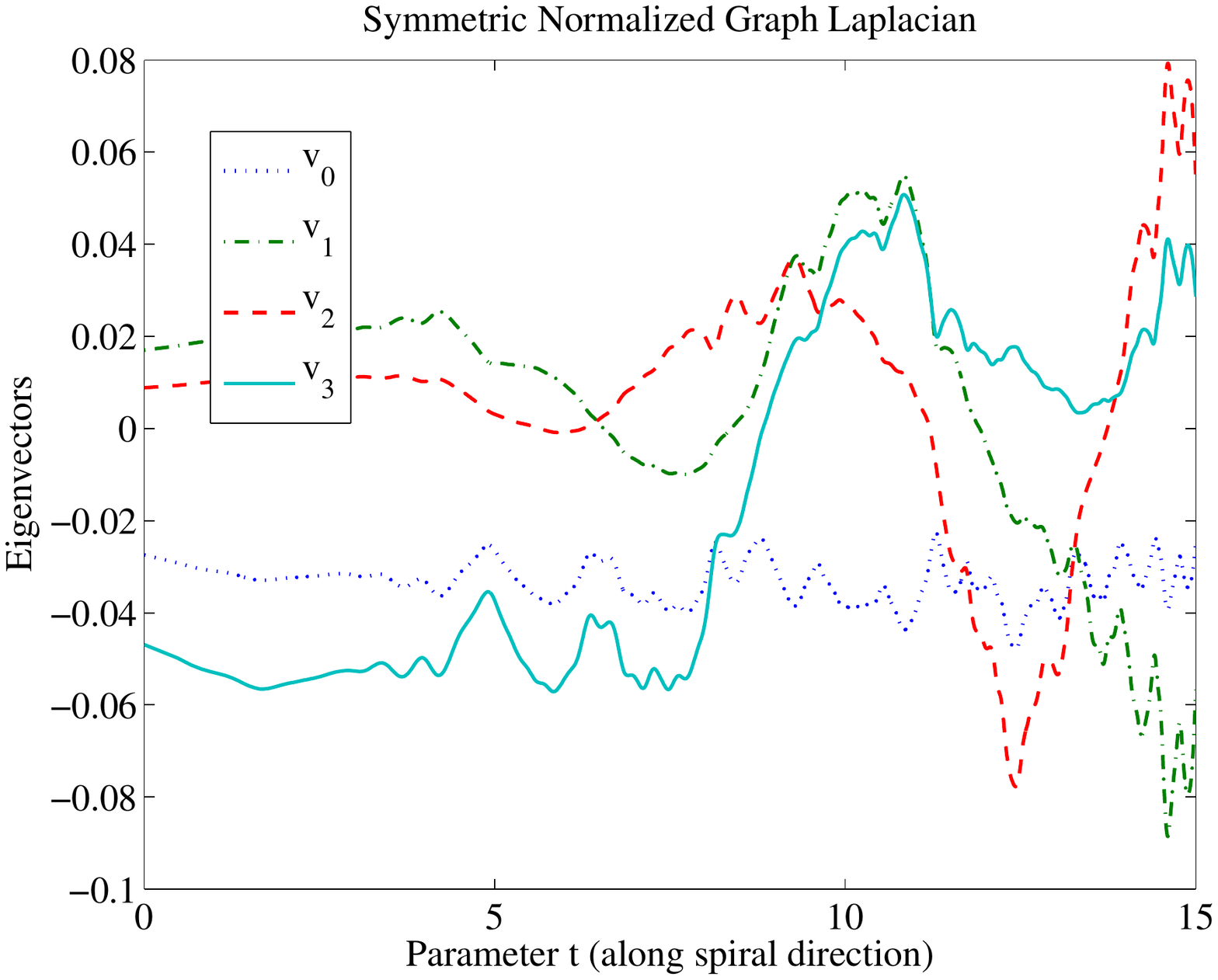}} &
{\hspace{-0.4in} \includegraphics[width=2.7in,angle=-0]{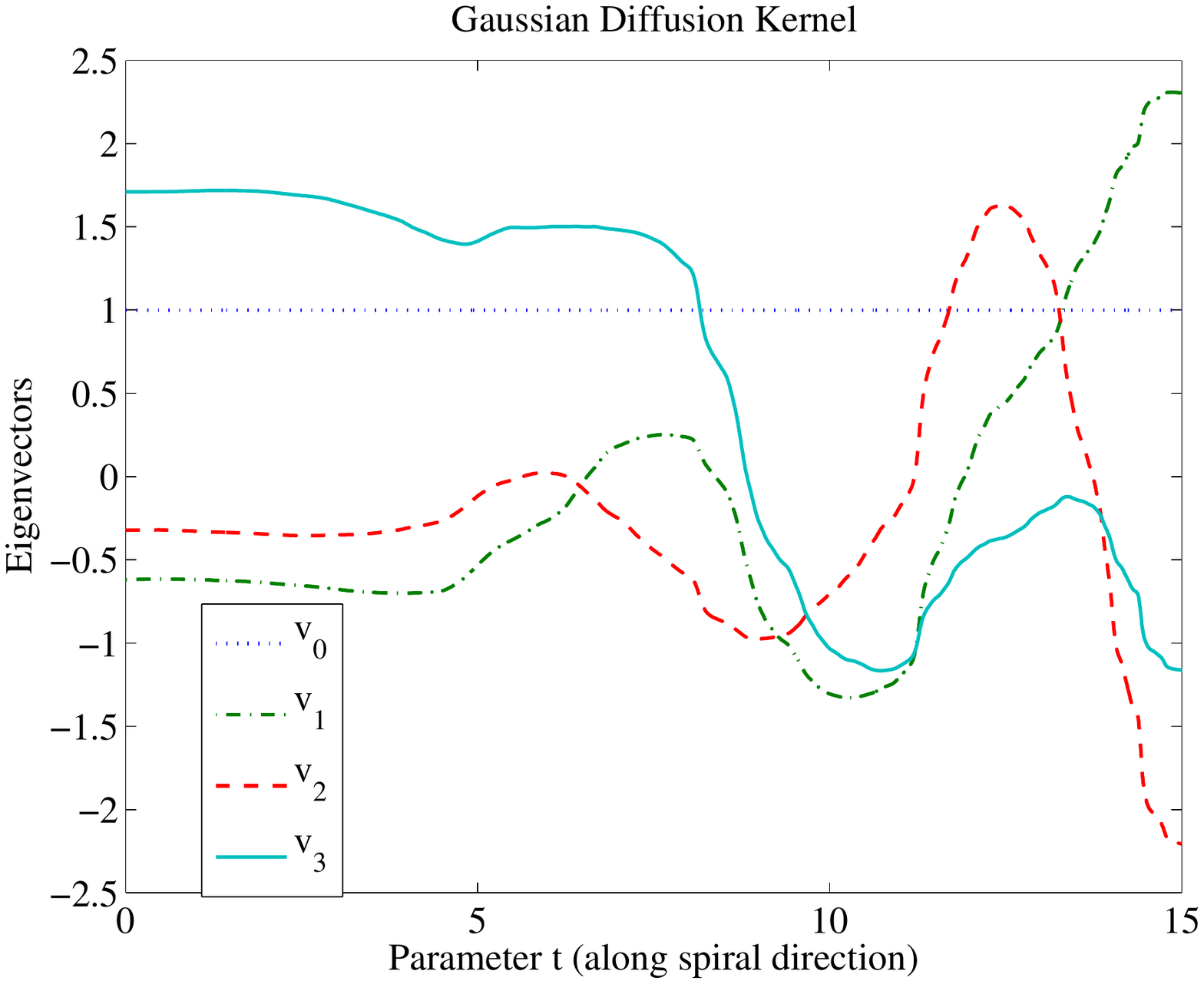}}
\end{tabular}
\end{center}
\vspace{-0.9in}
\caption{ \footnotesize Projection onto spiral direction  ($t=\sqrt{u}$) of the eigenvectors of the symmetric graph Laplacian, {\em left}, and of the Gaussian diffusion kernel, {\em right}. The latter estimates are less noisy.}
\label{fig::spiral_projections}
\end{figure}

\section{Theory}\label{sec::theory}
 In this section, we  derive 
  theoretical bounds on the loss (Eq.~\ref{eq:lossfunction}) of a series regression estimator with a radial kernel for a standard fixed RKHS setting (Theorem~\ref{theorem::loss_fixedkernel}), as well as a setting where the kernel bandwidth $\varepsilon_n$ {\em varies} with the sample size $n$ (Theorem~\ref{theorem::loss_epskernel}). We also further elaborate on the connection between spectral series and Fourier analysis by generalizing the well-known link between Sobolev differentiable signals and the approximation error in a Fourier basis. 

\comment{Here we investigate how the performance of a spectral series estimator with a radial kernel depends on the choice of the smoothing parameters  $J$ and $\varepsilon$. 
 We also address questions such as when the basis representation  $\hat{f}(x)$ in Eq.~\ref{eq::series_est} is sparse, and 
whether our method adapts to the intrinsic dimensionality of the data.}
 \comment{$L(f,\hat f)  =   \frac{1}{n} \sum_{i=1}^{n}  | f(X_i) - \hat{f} (X_i;\mathcal{D})|^2$
and the associated risk
$$R_n(f,\hat f) =  \frac{1}{n} \sum_{i=1}^{n}  \E_{\mathcal{D}}  | f(X_i) - \hat{f} (X_i;\mathcal{D})|^2, $$
where the expectation is averaging over data $\mathcal{D}=\{(X_1,Y_1), \ldots, (X_n,Y_n)\}$.}

Using the same notation as before, let
\begin{eqnarray*}
f(x) =  \sum_{j=0}^{\infty} \beta_{\varepsilon,j} \psi_{\varepsilon,j}(x), \  \  \ &f_{\varepsilon,J}(x) = \sum_{j=0}^{J}  \beta_{\varepsilon,j}  \psi_{\varepsilon,j}(x),& \\
&\hat{f}_{\varepsilon,J}(x) = \sum_{j=0}^{J}  \hat{\beta}_{\varepsilon,j}  \hat{\psi}_{\varepsilon,j}(x),& 
\end{eqnarray*}
where  $\beta_{\varepsilon,j} =
\int_\X f(x) \psi_{\varepsilon,j}(x) dS_\varepsilon(x)$ and  $\hat{\beta}_{\varepsilon,j}  =  \frac{1}{n} \sum_{i=1}^n Y_i \hat \psi_{\varepsilon,j}(X_i) \hat{s}_{\varepsilon}(X_i)$. 
We write
$$|f(x)-\hat{f}_{\varepsilon,J}(x)|^2 \leq 2|f(x)-f_{\varepsilon,J}(x)|^2  + 2|f_{\varepsilon,J}(x) - \hat{f}_{\varepsilon,J}(x)|^2,$$ and refer to the two terms as ``bias'' and ``variance''.  Hence,  we define the 
  integrated  bias and variance 
$$L_{\rm bias}  =  \int_{\X}   | f(x) - f_{\varepsilon,J}(x)|^2 dP(x),$$  and
$$ L_{\rm var}  =   \int_{\X}  |f_{\varepsilon,J}(x)- \hat{f}_{\varepsilon,J}(x)   |^2 dP(x),$$
\comment{$$L_{\rm bias}  \equiv  \frac{1}{n} \sum_{i=1}^{n}  | f(X_i) - f_{\varepsilon,J}(X_i)|^2,$$  and
$$ L_{\rm var}  \equiv \frac{1}{n} \sum_{i=1}^{n}  |  f_{\varepsilon,J}(X_i)- \hat{f}_{\varepsilon,J}(X_i)   |^2,$$} and bound the two components separately. 
Our assumptions are:

{\bf(A1)} $P$ has compact support $\X$ and bounded density $0<a \leq p(x) \leq b < \infty$,  $\forall x \in \X$.

{\bf(A2)} The weights are positive and bounded; that is, $\forall x, y \in \X$,
$$0 < m \leq k_{\varepsilon}(x,y) \leq M,$$ where $m$ and $M$ are constants that do not depend on $\varepsilon$.

{\bf(A3)}  The psd operator $A_\varepsilon$ has nondegenerate eigenvalues; i.e., $$1 \equiv \lambda_{\varepsilon,0} > \lambda_{\varepsilon,1} > \lambda_{\varepsilon,2} > \ldots \lambda_{\varepsilon,J} > 0.$$

{\bf(A4)} For all $0 \leq j \leq J$ and  $X \sim P$,  there exists some constant  $C < \infty$ (not depending on $n$) such that
$$\mathbb{E} \left[ |\hat{\varphi}_{\varepsilon,j}(X) - \varphi_{\varepsilon,j}(X)|^2 \right] < C,$$
where $\varphi_{\varepsilon,j}(x)= \psi_{\varepsilon,j}(x) s_{\varepsilon}(x)$ and $\hat \varphi_{\varepsilon,j}(x)=\hat \psi_{\varepsilon,j}(x) \widehat s_{\varepsilon}(x)$.  

Without loss of generality, we assume that the eigenfunctions $\psi_{\varepsilon,j}$ are estimated using an unlabeled sample 
$\widetilde{X}_1,\ldots,\widetilde{X}_n$ that is drawn independently from the data used to estimate the coefficients $\beta_{\varepsilon,j}$. This is to simplify the proofs and can always be achieved by splitting the data in two sets.

\subsection{Bias}\label{sec::bias}
A key point is that the approximation error of the regression depends on the smoothness of $f$ relative to $P$. Here we present two different calculations of the bias based on two related notions of smoothness.  The first notion is standard in the kernel literature and is based on RKHS norms. The second notion is based on our diffusion framework and can be seen as a generalization of Sobolev differentiability.

\subsubsection*{Method 1: Smoothness measured by RKHS norm.} Let $\displaystyle \tilde{a}_\varepsilon(x,y)=\frac{k_\varepsilon(x,y)}{\sqrt{p_\varepsilon(x)} \sqrt{p_\varepsilon(y)}},$ where $\varepsilon$ is a strictly positive number. Under previous assumptions, this kernel is symmetric and psd with a unique RKHS which we denote by $\mathcal{H}_\varepsilon$. A standard way to measure smoothness of a function $f$  in a RKHS $\mathcal{H}_\varepsilon$ is through the RKHS norm $\|f\|_{\mathcal{H}_\varepsilon}$ (see, e.g.,~\cite{Minh:2006}).  One can then define function classes $$\mathcal{H}_{\varepsilon,M} = \{f \in \mathcal{H}_\varepsilon : \|f\|_{\mathcal{H}_\varepsilon} \leq M\},$$ where $M$ is a positive number dependent on $\varepsilon$. 
\begin{proposition}\label{prop::bias_RKHS} Assume $f \in \mathcal{H}_{\varepsilon,M}$, where $\varepsilon>0$. Then,
$$L_{\rm{bias}} =  O(M \lambda_{\varepsilon,J}) .$$
\end{proposition}
For fixed $\varepsilon$, $\mathcal{H}_{\varepsilon,M}$ contains ``smoother'' functions for smaller values of $M$.  



\subsubsection*{Method 2: Smoothness measured by diffusion operator.}
Alternatively, let
\begin{equation}
G_{\varepsilon} = \frac{A_{\varepsilon}-I}{{\varepsilon} },  \label{eq::G_eps} 
 \end{equation}
where $I$ is the identity. 
The operator $ G_{\varepsilon} $ has the same eigenvectors  $\psi_{\varepsilon,j}$ as the differential operator $A_{\varepsilon}$. Its eigenvalues are given by 
$ -\nu_{\varepsilon,j}^2=\frac{\lambda_{\varepsilon,j}-1}{\varepsilon},$ where $\lambda_{\varepsilon,j}$ are the eigenvalues of $A_{\varepsilon}$.
Define the functional 
\begin{equation} \mathcal{J}_{\varepsilon}(f) = -\langle  G_{\varepsilon} f, f \rangle_{\varepsilon}\label{eq::smoothness_functional}\end{equation} 
which maps a function $f \in  L^2(\X,P)$ into a non-negative real number. For small $\varepsilon$,   $\mathcal{J}_{\varepsilon}(f)$ measures the variability of the function $f$ with respect to the distribution $P$. The expression is a variation of the graph Laplacian regularizers popular in semi-supervised learning~\citep{ZGL:2003}.  In fact, a Taylor expansion yields
$\displaystyle  G_\varepsilon f = -\triangle f +\frac{\nabla p}{p} \cdot \nabla f+ O(\varepsilon)$ where $\nabla$ is the gradient operator and  $\displaystyle \triangle  = - \sum_{j=1}^{d}\frac{\partial^2}{\partial x_j^2}$ is the psd Laplace operator in $\mathbb{R}^d$. In kernel regression smoothing, the extra term $\displaystyle \frac{\nabla p}{p} \cdot \nabla f$ is considered an undesirable extra bias, called design bias. In classical regression, it is removed by using local linear smoothing~\citep{Fan:1993}, which is asymptotically equivalent to replacing the original kernel $k_\varepsilon(x,y)$ by the bias-corrected kernel 
 $k_\varepsilon^{*}(x,y)=\frac{k_\varepsilon(x,y)}{{p_\varepsilon(x)} {p_\varepsilon(y)}}$~\citep{Coifman:Lafon:06}.


The following result
bounds the approximation error of an orthogonal series expansion of $f$.  The bound is consistent with Theorem~2 in~\cite{Zhou:Srebro:2011}, which applies to the more restrictive setting of SSL with infinite unlabeled data and $\varepsilon \rightarrow 0$.  Our result holds for {\em all} $\varepsilon$ and $J$ and  
  does {\em not} assume unlimited data. 
 
\begin{proposition}\label{proposition::bias_diffop}
For $f  \in   L^2(\X,P)$,
\begin{eqnarray} 
	&&\int_{\X}  | f(x) - f_{\varepsilon,J} (x) |^2 dS_{\varepsilon}(x) \ \leq \  \frac{ \mathcal{J}_{\varepsilon}(f)}{\nu^2_{\varepsilon,J+1}} \label{eq:bias_bound}\\
	 &&L_{\rm{bias}} = O \left(\frac{\mathcal{J}_{\varepsilon}(f)}{\nu^2_{\varepsilon,J+1}} \right),  \nonumber
\end{eqnarray}
where 
 $- \nu^2_{\varepsilon,J+1}$ is the $(J+1)^{th}$ eigenvalue of  $G_{\varepsilon}$.
\end{proposition}

\subsubsection*{Smoothness and Sparsity}  
In the limit $\varepsilon \rightarrow 0$, we have several interesting results, including a generalization of the classical connection between Sobolev differentiability and the error decay of Fourier approximations~\citep[Section 9.1.2]{Mallat:2009} to a setting with adaptive bases and {\em high-dimensional} data. We denote the quantities derived from  the bias-corrected kernel $k_\varepsilon^{*}$ by $A_{\varepsilon}^{*}$,  $G_{\varepsilon}^{*}$,  $\mathcal{J}_{\varepsilon}^{*}$ and so forth.

\begin{definition} \label{def::smoothnessP}{\bf (Smoothness relative to P)} A function $f$ is smooth relative to $P$ if  $$\int_{\X} \|\nabla f(x)\|^2 dS(x) \leq c^2 < \infty,$$ where $S(A) = \frac{\int_A p(x) dP(x)}{\int p(x) dP(x)}$ is the stationary distribution of the random walk on the data as $\varepsilon \rightarrow 0$. The smaller the value of $c$, the smoother the function.\end{definition}

\begin{lemma}\label{lemma:conv_smoothfun} For functions $f \in C^3(\X)$ whose gradients vanish at the boundary,
$$ \lim_{\varepsilon \rightarrow 0} \mathcal{J}^{*}_{\varepsilon}(f)  = \int_{\X}  \|\nabla f(x) \|^2 dS(x).$$
\end{lemma}
This is similar to the convergence of the (un-normalized) graph Laplacian regularizer to the density-dependent smoothness functional $\displaystyle \int_\X \| \nabla f(x)\|^2 p^2(x) dx$~\citep{BousquetCH03}.


Next we will see that {\em smoothness} relative to $P$ (Definition~\ref{def::smoothnessP}) and {\em sparsity} (with respect to the $L^2$ norm) in the eigenbasis of the diffusion operator  (Definition~\ref{def::sparsityB} below) are really the same thing.  As a result, we can link smoothness and sparsity to the rate of the error decay of the eigenbasis approximation.

\begin{definition} {\bf (Sparsity in $\mathcal{B}$)}\label{def::sparsityB}
A set of real numbers $\theta_1, \theta_2, \ldots$ lies in a Sobolev ellipsoid $\Theta(s,c)$ if $\sum_{j=1}^{\infty} j^{2s} \theta_{(j)}^2 \leq c^2$ for some number $s>1/2$.
For a given basis $\mathcal{B}=\{\psi_1, \psi_2, \ldots \}$, let
$$ W_\mathcal{B}(s,c) = \left\{ f=\sum_j \beta_j \psi_j : \beta_1, \beta_2, \ldots \in \Theta(s,c) \right\}$$
where $s>1/2$. Functions in $W_\mathcal{B}(s,c)$ are sparse in $\mathcal{B}$. The larger the value of $s$, the sparser the representation.
\end{definition}

\begin{theorem}\label{prop::errordecay} Assume that $\mathcal{B}=\{\psi_1, \psi_2, \ldots \}$ are the eigenvectors of $\triangle$ with eigenvalues $\nu_j^2=O(j^{2s})$ for some $s>1/2$. Let $f_J(x)=\sum_{j \leq J} \beta_j \psi_j(x)$. Then, the following two statements are equivalent:
\begin{enumerate}
 \item   $\int_{\X} \|\nabla f(x)\|^2 dS(x) \leq c^2 \ $ (smoothness relative to P)
 \item  $\displaystyle f \in  W_\mathcal{B}(s,c) \  \ $ (sparsity in $\mathcal{B}$). 
 \end{enumerate}
 Furthermore, sparsity in $\mathcal{B}$ (or smoothness relative to P) implies
  $$\displaystyle \int_{\X}  | f(x) - f_J (x) |^2 dS(x) = o \left(\frac{1}{J^{2s}}\right) .$$
\end{theorem}
The rate $s$ of the error decay depends on the dimension of the data. We will address this issue in Sec.~\ref{sec:total_loss}.

\subsection{Variance}
The matrix $\mathbb{A}_\varepsilon$ (defined in Eq.~\ref{eq::A_matrix}) can be viewed as a perturbation of the integral operator $A_\varepsilon$ due to finite sampling.  To estimate the variance, we bound the difference $\psi_{\varepsilon,j} - \widehat{\psi}_{\varepsilon,j}$, where  $\psi_{\varepsilon,j}$ are the eigenvectors of  $A_\varepsilon$, and $\widehat{\psi}_{\varepsilon,j}$ are the Nystr\"{o}m extensions  (Eq.~\ref{eq::nystrom}) of the eigenvectors of $\mathbb{A}_\varepsilon$. 
 We adopt a strategy from Rosasco et al.~\citep{Rosasco:EtAl:2010}, which is to introduce two new integral operators that are related to  $A_\varepsilon$ and $\mathbb{A}_\varepsilon$ but  both act on an auxiliary\footnote{This auxiliary space only enters the intermediate derivations and plays no role in the error analysis of the algorithm itself.} RKHS $\mathcal{H}$  of smooth functions (see Appendix~\ref{appendix::variance} for details).     As before, we write $\varepsilon_n$ to indicate that we let the kernel bandwidth $\varepsilon$ depend on the sample size $n$.  
 
\begin{proposition}  \label{prop::eigenvec_error} 
Let $\varepsilon_n\to 0$ and $n \varepsilon_n^{d/2}/
\log(1/\varepsilon_n) \to \infty$ as $n \rightarrow 0$.
 Under assumptions (A1)-(A4) and $\forall \ 0 \leq j \leq J$,
$$\|\psi_{\varepsilon,j} - \hat{\psi}_{\varepsilon,j}\|_{L^2(\mathcal{X},P)} = O_P\left(\frac{\gamma_n}{\delta_{\varepsilon,j}}\right),$$
where  $\gamma_n = \sqrt{\frac{ \log(1/\varepsilon_n)}{n \varepsilon_n^{d/2}}}$ and 
 $\delta_{\varepsilon,j}=\lambda_{\varepsilon,j}-\lambda_{\varepsilon,j+1}.$
\end{proposition}

\begin{proposition}\label{prop::variance} 
 Let $\varepsilon_n\to 0$ and $n \varepsilon_n^{d/2}/
\log(1/\varepsilon_n) \to \infty$. Under (A1)-(A4) and for $f \in   L^2(\X,P)$,   it holds that
 $$L_{\rm{var}} 
= J \left( O_P \left( \frac{1}{n}\right)  + O_P\left(\frac{\gamma_n^2}{\Delta_{\varepsilon,J}^2}\right) \right).$$
 where  $\Delta_{\varepsilon,J}=\min_{0\leq j \leq J} (\lambda_{\varepsilon,j}-\lambda_{\varepsilon,j+1})$. 
 \end{proposition}
 
\comment{
\begin{remark} The second term in Lemma~\ref{lemma::betahat} and Proposition~\ref{prop::variance}  is due to the sample error in estimating the eigenvectors of $A_\varepsilon$. In a semi-supervised learning setting where the eigenvectors are derived from an infinite number of unlabeled data, this estimation error vanishes, yielding  much faster rates:
$|\hat{\beta}_{\varepsilon,j}-\beta_{\varepsilon,j}|^2 = O_P\left(\frac{1}{n}\right)$ and $L_{\rm{var}} = J O_P \left( \frac{1}{n}\right)$. 
\end{remark}
}

\subsection{Total Loss}\label{sec:total_loss}
 \subsubsection*{Fixed Kernel} In kernel machine learning, it is standard to assume a {\em fixed} RKHS $\mathcal{H}_k$, e.g., with norm $\|\cdot\|_{\mathcal{H}_k}$ and a fixed kernel $k$ with a bandwidth $\varepsilon$ not dependent on $n$.
 From Propositions~\ref{prop::bias_RKHS} and~\ref{prop::variance}    and under assumptions (A1)-(A4), we then have the following result: 
\begin{theorem}\label{theorem::loss_fixedkernel}
Assume $f \in \mathcal{H}_k$ with finite norm; i.e., $\|f\|_{\mathcal{H}_k} \leq M$ for some constant $M < \infty$. Then, 
\begin{equation}\label{eq:loss_fixed_eps}
L(f,\hat f) =  O(\lambda_{\varepsilon,J})  + J   O_P \left( \frac{1}{n}\right)  + J O_P\left(\frac{1}{n\Delta_{\varepsilon,J}^2}\right),
\end{equation}
where  $\Delta_{\varepsilon,J}=\min_{0\leq j \leq J} (\lambda_{\varepsilon,j}-\lambda_{\varepsilon,j+1})$.
 \end{theorem}
The problem is that $\mathcal{H}_k$, $M$,  and the eigenvalues $\lambda_{\varepsilon,j}$, all depend on $\varepsilon$.
This dependence is complicated and poorly understood.
Hence, in what follows, we will instead of the RKHS norm use an alternative measure of smoothness based on the diffusion operator (Method 2 in Sec.~\ref{sec::bias}).  This simplifies the theory and will allow us to analyze the dependence of the series estimator on tuning parameters and sparse structure. 

\subsubsection*{Kernel with Decreasing Bandwidth}
Consider now a Gaussian kernel $k_{\varepsilon}$ with a bandwidth $\varepsilon_n$ that {\em decreases} with increasing $n$.
From Propositions~\ref{proposition::bias_diffop} and~\ref{prop::variance}   and under assumptions (A1)-(A4), we have the following results:
\begin{theorem}\label{theorem::loss_epskernel}
Let $\varepsilon_n\to 0$ and $n \varepsilon_n^{d/2}/
\log(1/\varepsilon_n) \to \infty$ as $n \rightarrow \infty$. Then, for $f \in   L^2(\X,P)$,
\begin{equation}\label{eq:loss_vary_eps}
L(f,\hat f)  =   O \left(\frac{\mathcal{J}_{\varepsilon}(f)}{\nu^2_{\varepsilon,J+1}} \right) + J O_P \left( \frac{1}{n}\right)  + J O_P\left(\frac{\gamma_n^2}{\Delta_{\varepsilon,J}^2}\right),
\end{equation}
where  $\mathcal{J}_{\varepsilon}(f) = -\langle  G_{\varepsilon} f, f \rangle_{\varepsilon}$, $\nu^2_{\varepsilon,J+1}$ is the $(J+1)^{th}$ eigenvalue of  $-G_{\varepsilon}$,   $\gamma_n = \sqrt{\frac{ \log(1/\varepsilon_n)}{n \varepsilon_n^{d/2}}}$, and $\Delta_{\varepsilon,J}=\min_{0\leq j \leq J} (\lambda_{\varepsilon,j}-\lambda_{\varepsilon,j+1})$.
 \end{theorem} 

\begin{corollary}\label{corollary::loss_epskernel}
Assume that $f \in   C_b^3(\mathcal{X})$ and that the kernel $k=k_\varepsilon^{*}$ is corrected for bias. 
Then, for $\varepsilon_n\to 0$ and $n \varepsilon_n^{d/2}/
\log(1/\varepsilon_n) \to \infty$,
\begin{equation}\label{eq:loss_vary_eps}
L(f,\hat f)  =   \frac{\mathcal{J}(f)O(1)+O(\varepsilon_n)}{\nu^2_{J+1}}  + J O_P \left( \frac{1}{n}\right)  + JO_P\left(\frac{\gamma_n^2}{\varepsilon_n \Delta_{J}^2}\right),
\end{equation}
where $\nu^2_{J+1}$ is the $(J+1)^{th}$ eigenvalue of  $\triangle$,   $\mathcal{J}(f) =  \int_{\X}  \|\nabla f(x) \|^2 dS(x)$, and $\Delta_J=\min_{0\leq j \leq J} (\nu_{j+1}^2-\nu_j^2)$.
 \end{corollary} 
 Some comments on these results:
The first term in Eqs.~\ref{eq:loss_fixed_eps}-\ref{eq:loss_vary_eps} corresponds to the approximation error of the estimator and decays with $J$. The second and third terms correspond to the variance. Note that the variance term $J O_P \left( \frac{1}{n}\right)$ is the same as the variance of a {\em traditional} orthogonal series estimator in {\em one} dimension only; in $d$ dimensions, the variance term for a traditional tensor product basis is $O_P \left( \frac{1}{n}\right) \prod_{i=1}^{d} J_i$ where $J_i$ is the number of components in the $i$th direction~\citep{efromovich}. Hence, there is a considerable gain in using an adaptive bias,  but we incur an additional variance term  $J O_P\left(\frac{\gamma_n^2}{\varepsilon_n \Delta_{J}^2}\right)$ from estimating the basis.\footnote{In an SSL setting (Remark~\ref{remark::SSL}), this extra estimation error vanishes in the limit of infinite unlabeled data.}

If we balance the  two $\varepsilon$-terms in Eq.~\ref{eq:loss_vary_eps}, we get a bandwidth of $\varepsilon_n  \asymp (1/n)^{2/(d+4)}$. With this choice of $\varepsilon_n$ 
 and by ignoring terms of lower order,
the rate becomes
\begin{equation}
L(f,\hat f)  =   O\left(\frac{\mathcal{J}(f)}{\nu^2_{J+1}}\right)   + \frac{J}{\Delta_J^2}O_P\left(\frac{\log n}{n}\right)^{\frac{2}{d+4}}.
\end{equation}
Finally, if we apply the results in \cite{Coifman:Lafon:06,Gine:Guillou:02,Gine:Koltchinskii:2006,Rosasco:EtAl:2010} to general Riemannian manifolds (see, for example, \cite{Henry:Rodriguez:2009,Ozakin:Gray:2009, BerrySauer:2016}   for kernel density estimation on manifolds), and use that the eigenvalues of the Laplace-Beltrami operator on an $r$-dimensional Riemannian manifold are $\nu_j^2 \sim j^{2/r}$~\citep{SafarovVassiliev96}, we obtain the following corollary:

\begin{corollary}\label{corollary:manifold_example}
Suppose the support of the data is on a compact $C^\infty$ submanifold of $\mathbb{R}^d$ with intrinsic dimension $r$, and suppose that $f$ is smooth relative to P (Definition~\ref{def::smoothnessP}). 
Under the assumptions of Theorem \ref{theorem::loss_epskernel}
and Corollary \ref{corollary::loss_epskernel}, we obtain the rate
$$L(f,\hat f)  =   O \left(\frac{1}{J^{2/r}}\right) + J^{2\left(1-\frac{1}{r}\right)} O_P\left(\frac{\log n}{n}\right)^{\frac{2}{r+4}}.$$
It is then optimal to take $J \asymp (n/ \log n)^{\frac{1}{r+4}}$, in which case the upper bound becomes
$$\left(\frac{\log n}{n}\right)^{\frac{2}{(r+4)r}}.$$
\end{corollary}

\comment{
\begin{remark}
On the other hand, if one ignores the estimation error of the basis as in, e.g.,~\cite{Zhou:Srebro:2011} and~\cite{Nadler_semi}, then we 
can take $\varepsilon$ to be arbitrarily close to zero, and we instead have
\begin{equation}\label{eq:loss_vary_eps_SSL}
L(f,\hat f)  =  O \left(\frac{1}{J^{2/d}}\right) + J O_P \left( \frac{1}{n}\right)
\end{equation}
This loss is minimized by taking $J \asymp n^{d/(d+2)}$, yielding the rate
$$\frac{1}{n^{\frac{2}{2+d}}},$$
which is, up to a logarithmic term, in agreement with \cite{Zhou:Srebro:2011}. 
\end{remark}
}

We make the following observations for a spectral series estimator with flexible kernel bandwidth:

\begin{enumerate}[(i)]

\item  {\bf Adaptiveness to Low-Dimensional Structure.} If the data in $\mathbb{R}^d$ has intrinsic dimension $r \ll d$, then the rate  $n^{-1/O(r^2)}$ above is a significant improvement of the minimax rate $n^{-1/O(d)}$ for a nonparametric regressor in  $\mathbb{R}^d$. Our estimator automatically adapts to sparse structure and does not require the knowledge of $r$ or an estimated $r$   in practice. The optimal error rate is achieved when the smoothing parameters $J$ and $\varepsilon$ are properly selected for the given $r$, and the amount of smoothing is in practice chosen by cross-validation as in Sec.~\ref{sec::loss_function}. 

\item {\bf Minimax Optimality.} In a semi-supervised learning setting, 
 the estimation error of the basis vanishes in the limit of  infinite unlabeled data.
 The loss then reduces to
\begin{equation}\label{eq:loss_vary_eps_SSL}
L(f,\hat f)  =  O \left(\frac{1}{J^{2/r}}\right) + J O_P \left( \frac{1}{n}\right) ,
\end{equation}
which is minimized by taking $J \asymp n^{r/(r+2)}$. At the minimum, we achieve the rate
$$\frac{1}{n^{\frac{2}{2+r}}},$$
the minimax rate for a nonparametric estimator of Sobolev smoothness $\beta=1$ in $\mathbb{R}^D$, where $D=r$.
The latter result is also, up to a logarithmic term, in agreement with \cite{Zhou:Srebro:2011}. 
\end{enumerate}


 \comment{
\begin{remark} The last term in Theorem~\ref{theorem::loss_epskernel}, as well as in Theorem~\ref{theorem::loss_fixedkernel}, 
is due to the sample error in estimating the eigenvectors of $A_\varepsilon$. In a semi-supervised learning setting where the eigenvectors are derived from an infinite number of unlabeled data, this estimation error vanishes, yielding considerably faster rates with a variance contribution of
$L_{\rm{var}} = J O_P \left( \frac{1}{n}\right)$. 
\end{remark}}
 
\section{Numerical Examples}\label{sec:examples}
 Finally, we use data with complex dependencies to compare  the spectral series approach with classical kernel smoothing, k-nearest neighbors (kNN) regression, regularization in RKHS, and recent state-of-the-art manifold and local regression methods. 
 
 
 In our experiments, we split the data into three sets for training, validation, and testing, respectively. For the manifold regression estimators from Aswani et al \cite{Aswani:EtAl:2011} and Cheng et al. \cite{cheng2012local}, 
  we use the authors' codes with built-in cross-validation. For all other estimators, we tune parameters according to Sec.~\ref{sec::loss_function}. To assess the final models, we compute the estimated loss  $\hat{L}$ and standard error on the test data.\footnote{The estimated standard error of $\hat{L}$ is $s/\sqrt{n}$, where $s^2$ is the empirical variance of $(Y_i-\hat{f}(X_i))^2$ on the test set.}  

\subsection{Estimating Pose Using Images of Faces}\label{sec:isomap}
 
 In our first example, we consider images of artificial faces from the Isomap database~\cite{Tenenbaum:2000}.\footnote{\url{www.isomap.stanford.edu/datasets.html}} There are a total of $n=698$ $64 \times 64$ gray-scale images rendered with different orientation and lighting directions. Fig.~\ref{fig::diffMapISO} shows a visualization of these data where we use the first two non-trivial eigenvectors of the Gaussian diffusion kernel as coordinates  (i.e., Eq.~\ref{eq:eigenmap} with the approximate eigenvectors from Eq.~\ref{eq::nystrom}). 
 
 Our goal is to estimate the horizontal left-right pose of each face. We compare several different approaches to regression:  
 
(i) As a baseline, we choose the classical  Nadaraya-Watson estimator with a Gaussian smoothing kernel  (\emph{NW})  and the k-nearest neighbors regression estimator (\emph{kNN}). The latter estimator is known to be minimax optimal with respect to local intrinsic dimension~\cite{kpotufe2011k}.

  (ii) For the spectral series method (\emph{Series}), we implement the Gaussian kernel (\emph{Series-radial}) and polynomial kernels  $k(x,y)=(\langle x,y\rangle + 1)^q$ of different degrees $q$. We treat $q$ as a tuning parameter and we denote the  polynomial kernel with the smallest estimated loss by \emph{Series-polyBest}.   Note that choosing $q=1$ (\emph{Series-poly1})   is equivalent to a linear regression on eigenvectors computed with PCA.
    
   (iii) We also implement the RKHS method in Sec.~\ref{sec:kernel_ML} for the same set of kernels as \emph{Series}. 
 For a squared-error loss, Eq.~\ref{eq::regularization}  reduces to an infinite-dimensional, generalized ridge regression problem~\citep[Section 5.8.2]{Hastie:EtAl:09}. Hence, we use the term kernel ridge regression (\emph{KRR}) and denote the estimators by \emph{KRR-radial} and \emph{KRR-poly}.
 
(iv) The last group of estimators include recent
 manifold and local regression methods 
 \citep{Aswani:EtAl:2011,cheng2012local}\footnote{For code, go to \url{www.eecs.berkeley.edu/~aaswani/EDE_Code.zip} and \url{http://www.math.princeton.edu/~hauwu/regression.zip}.}: \emph{locOLS} is a local ordinary least squares, \emph{locRR} is a local ridge regression, \emph{locEN} is a local elastic net,  \emph{locPLS} is a local partial least squares,
 \emph{locPCR} is a local principal components regression, 
\emph{NEDE} is the 
nonparametric exterior derivative estimator,
\emph{NALEDE} is the 
nonparametric adaptive lasso exterior derivative estimator,
 \emph{NEDEP} is the 
nonparametric exterior derivative estimator for the ``large p, small n'' case, and
 \emph{NALEDEP} is the nonparametric adaptive
lasso exterior derivative estimator for the ``large p, small n'' case. The last 4 regression estimators (\emph{NEDE, NALEDE, NEDEP, NALEDP}) pose the regression as a least-squares problem with a term that penalizes for the regression vector lying in directions perpendicular to an estimated manifold; see \citep{Aswani:EtAl:2011} for details. In our comparison, we also include \emph{MALLER} \cite{cheng2012local} which
 first estimates the local dimension of the data and then performs local linear regression on a tangent plane estimate.

Manifold and local regression methods, unlike {\em Series}, quickly become computationally intractable in high dimensions. Hence, to be able to compare the different methods, we follow Aswani et al.~\cite{Aswani:EtAl:2011} and rescale the Isomap images from 
  from $64 \times 64$ down to $7 \times 7$ pixels in size. This reduces the number of covariates from $d=4096$ to $d=49$. In other words, we regress the left-right pose $Y$ (our response) on the rescaled image $X \in \mathbb{R}^{49}$ (our predictor). 
 We use  50\% of the data for training, 25\% for validation and 25\% for testing. All covariates are normalized to have mean 0 and standard deviation 1.

   Table \ref{tab::faces} and Fig.~\ref{fig::LossesISO} summarize the results of the final (cross-validated) estimators. The approaches that have best performance are \emph{Series-radial} and \emph{KRR-radial}. As expected, \emph{Series} and \emph{KRR} estimators yield similar losses. A first-order polynomial kernel, i.e., a global principal component regression with \emph{Series- or KRR-poly1},  performs worse than \emph{NW} and \emph{kNN}. Higher-order polynomial kernels (with degree $q=2$ resulting in the smallest loss) as well as  the manifold and local regression estimators (in particularly, \emph{NEDE} and \emph{MALLER}) improve the \emph{NW} and \emph{kNN} results but \emph{Series-radial} and \emph{KRR-radial} are still the best choices in terms of  statistical  and computational performance.  
\begin{table}[H]
 \begin{center}
\caption{ \footnotesize Estimated loss for Isomap face data.}
\label{tab::faces}
{ \footnotesize \begin{tabular}{c|c}
\hline
 Method  &  Loss (SE) \\ \hline
 {NW} & 1.71 (0.23) \\  
 {kNN} & 1.74 (0.21) \\ \hline
 {Series}-poly1& 2.96 (0.40)\\
      {Series}-polyBest (q=2) &0.22 (0.04)\\
  \color{red}  {Series}-radial&  \color{red} 0.16 (0.04)\\ \hline
 { KRR}-poly1& 2.95 (0.41) \\
  { KRR}-polyBest (q=2)& 0.22 (0.05)  \\
   \color{red} {KRR}-radial& \color{red} 0.15 (0.04)\\ \hline
\end{tabular}}
\qquad \qquad
{ \footnotesize \begin{tabular}{c|c}
\hline
Method  &  Loss (SE) \\ \hline
 {locOLS}& 0.65 (0.17)  \\ \hline
  {locRR}& 0.46 (0.16)  \\ \hline
  {locEN} & 0.47 (0.16)  \\ \hline
  {locPLS} & 0.65 (0.21)  \\ \hline
  {locPCR} & 0.95 (0.20)  \\ \hline
 {NEDE} & 0.44 (0.14)  \\ \hline
  {NALEDE}& 0.46 (0.14)  \\ \hline
  {NEDEP}& 0.81	 (0.31) \\  \hline
  {NALEDEP}& 0.85	 (0.33)  \\  \hline
   {MALLER}&  0.24     (0.06)  \\  \hline
\end{tabular}}
\end{center} 
\end{table}
\begin{figure}[H]
  \centering
\includegraphics[width=3.5in,angle=-0]{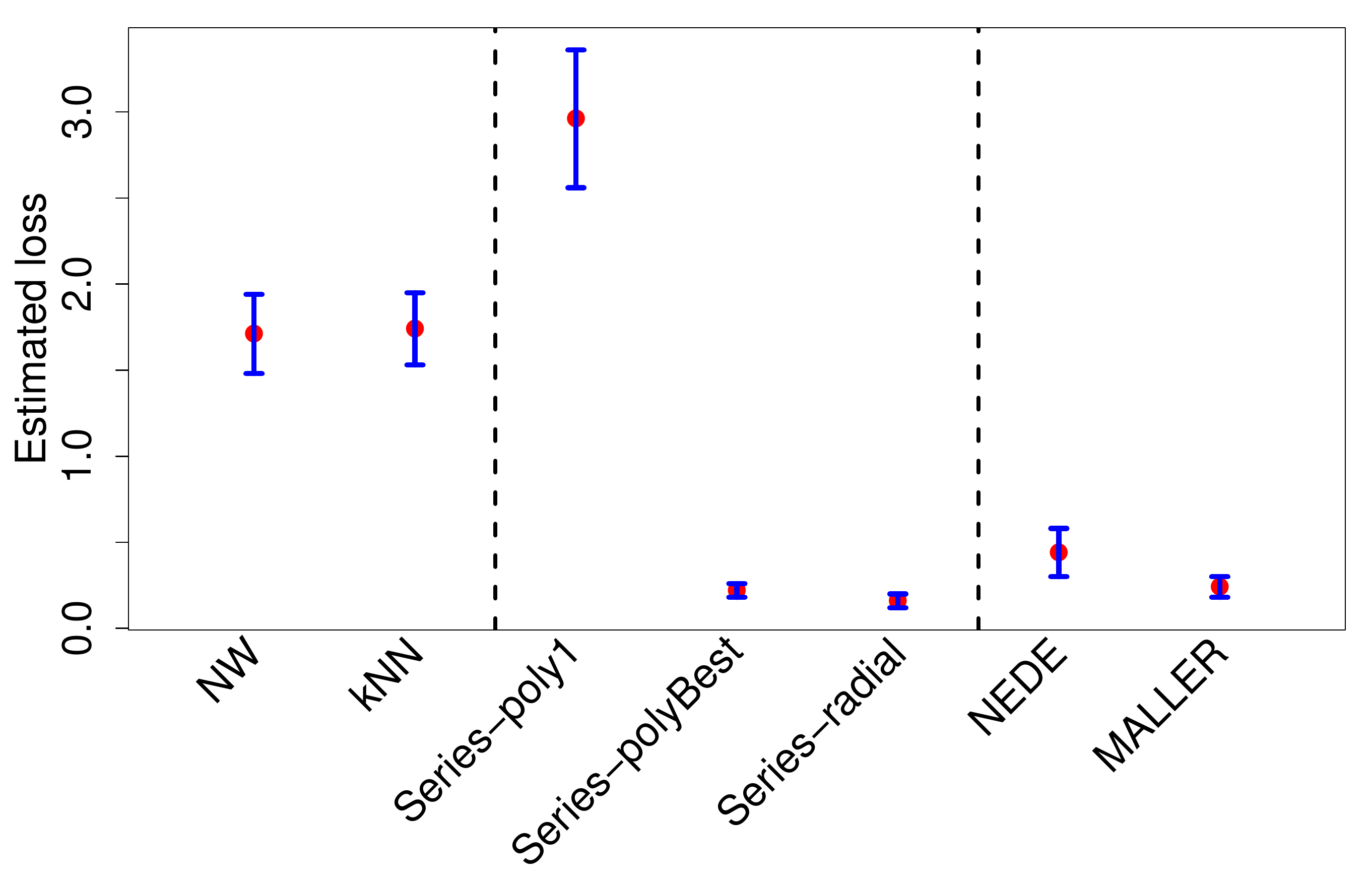}
  \caption{ \footnotesize Estimated loss of estimators for Isomap face data; see Table~\ref{tab::faces}.  For visibility, we have divided the estimators into three groups: classical NW kernel and kNN smoothers ({\em left}),  Series/RKHS-type estimators  ({\em center}), and manifold regression estimators, such as \emph{NEDE} and \emph{MALLER} ({\em right}). Bars represent standard errors.  }
  \label{fig::LossesISO}
\end{figure}

\subsection{Estimating Redshift Using SDSS Galaxy Spectra}\label{eq:SDSS_spectra}

In the following (high-dimensional) example, we predict the redshift of galaxies from high-resolution measurements of their emission spectra. 
Our initial data sample consists of galaxy spectra 
  from ten arbitrarily chosen spectroscopic plates of SDSS DR6.\footnote{\url{http://www.sdss.org/dr6/algorithms/redshift_type.html}} We preprocess and remove spectra according to the three cuts described in~\cite{Richards:EtAl:2009}. The final sample consists of $n=2812$ high-resolution spectra with flux measurements at $d=3501$ wavelengths. We renormalize each spectra so that it has unit norm. Our goal is to predict a galaxy's redshift $Y$ 
  where the predictor is an {\em entire} spectrum $\x \in \mathbb{R}^{3501}$.   Fig.~\ref{fig::diffusionMapSpec}a shows an example of a SDSS spectrum. Fig.~\ref{fig::diffusionMapSpec}b shows a low-dimensional visualization of the full data set when using the first few vectors of the diffusion basis as coordinates. 
  Each point in the plot represents a galaxy, and the color codes for the SDSS spectroscopic redshift. The redshift (the response $Y$) appears to vary smoothly with the eigencoordinates.
\begin{figure}[H]
  \begin{center}
  \subfloat[]{    \includegraphics[width=2.3in,angle=-0]{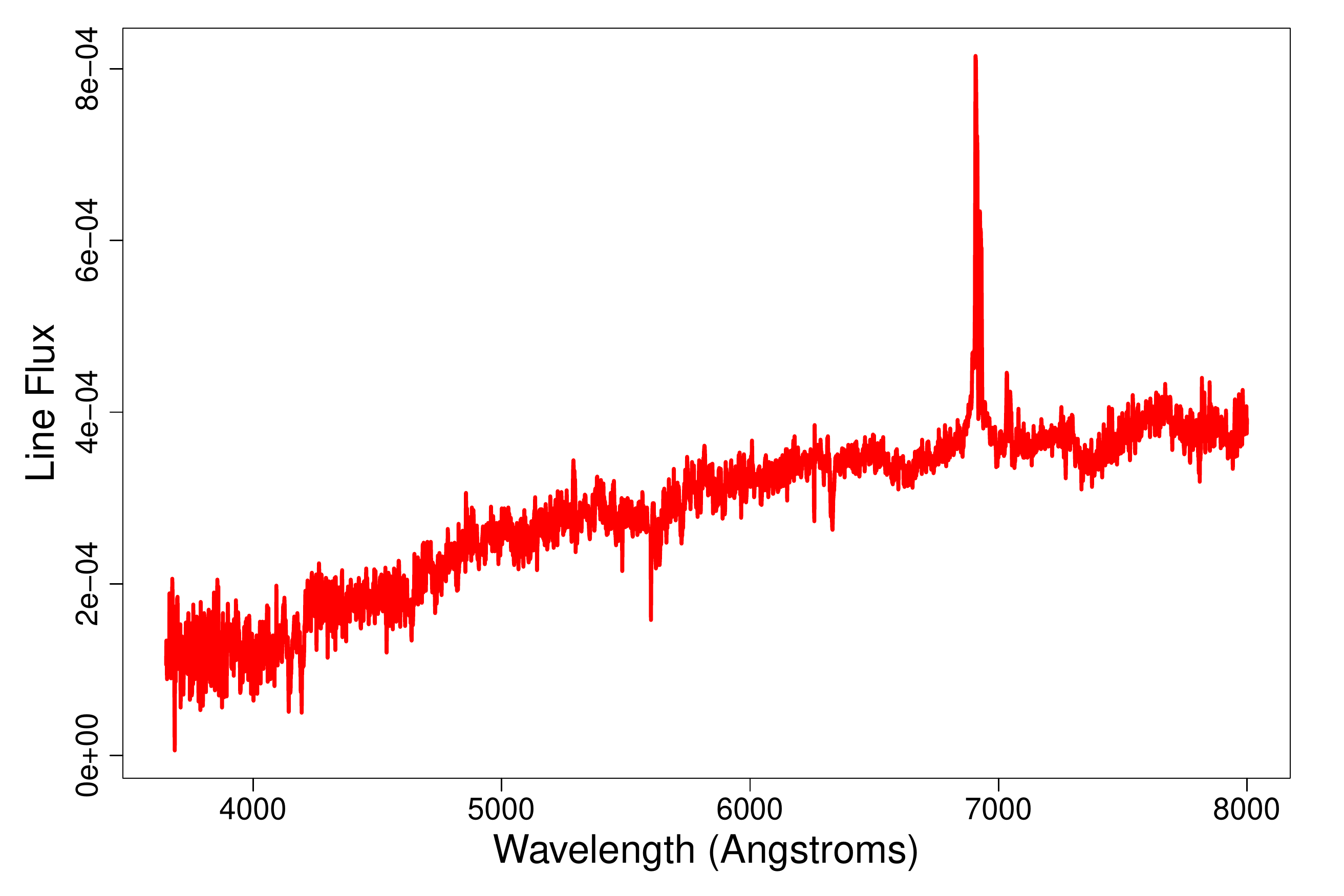} }
\subfloat[]{    \includegraphics[width=2.5in,angle=-0]{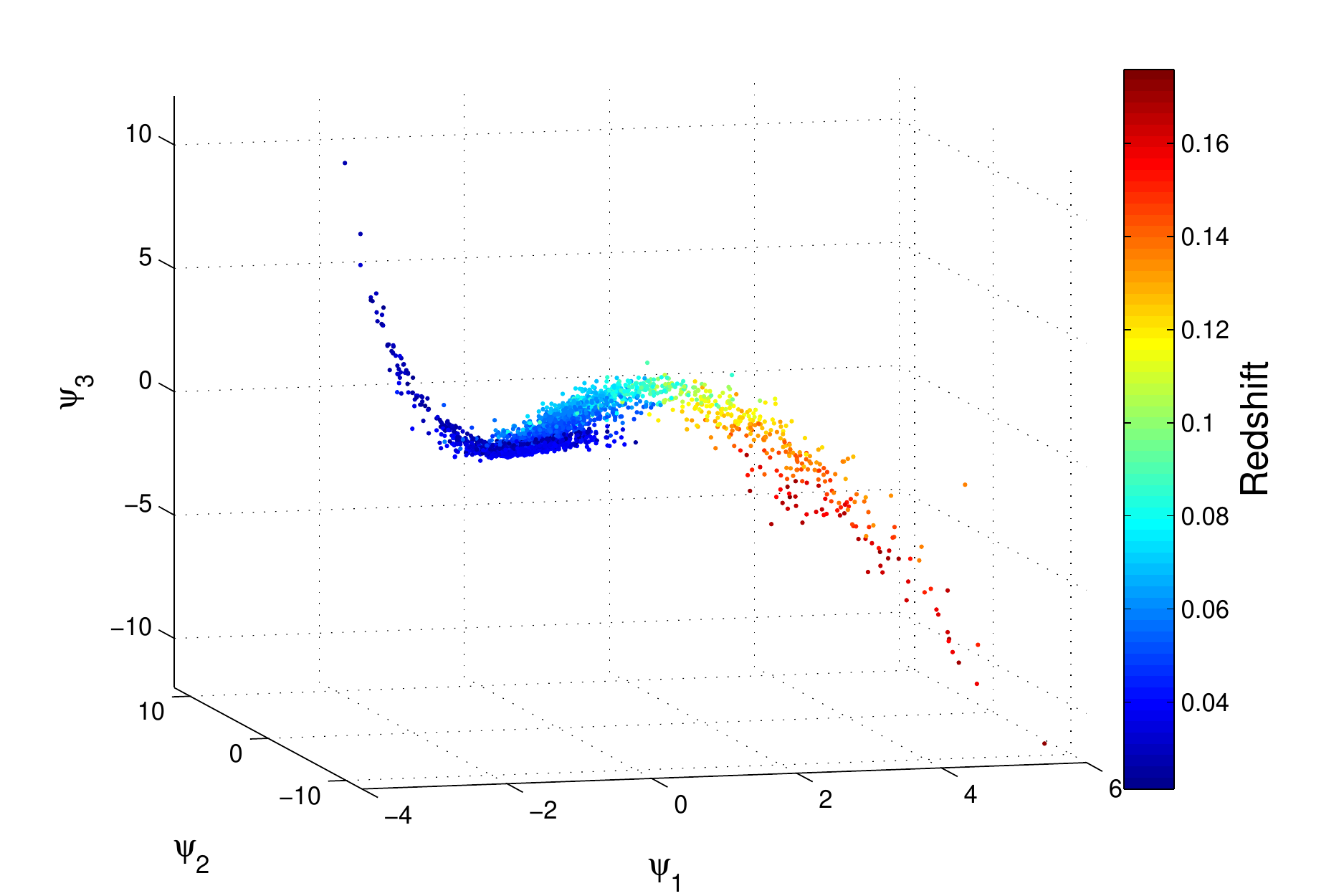} }
  \end{center}
  \caption{ \footnotesize Left: Example of a SDSS galaxy spectrum. Right: Embedding of a sample of SDSS galaxy spectra using the first three non-trivial eigenvectors of the Gaussian diffusion kernel as coordinates; the color codes for the true redshift.}
  \label{fig::diffusionMapSpec}
\end{figure}

For the regression, we use 50\% of the data for training, 25\% for validation and 25\% for testing. Due to the high dimension of the predictor ($d=3501$), we are unable to implement the computationally intensive manifold and local regression estimators from~\cite{Aswani:EtAl:2011}.  
 Table~\ref{tab::spec} and Fig.~\ref{fig::LossesSpec} summarize the results for the other approaches to regression. {\em Series} and {\em KRR} are essentially equivalent in terms of performance, and as before, the radial kernel ({\em Series-radial} and {\em KRR-radial}) yields the smallest estimated loss.   For these data, a linear dimensionality reduction with PCA (\emph{series-poly1}) improves upon the \emph{NW} and {\em kNN} regression results. {\em MALLER} and higher-order polynomials (with degrees 5 and 6) perform better than PCA, but {\em Series-radial} still has the smallest estimated loss.
  Moreover, {\em MALLER} is much slower than {\em Series}: the former estimator takes 34 minutes on a 2.70GHz Intel Core i7-4800MQ,  whereas {\em Series} with cross-validation takes less than a minute. 

 \begin{table}[H]
\begin{center}
\caption{ \footnotesize Estimated loss for redshift prediction using SDSS galaxy spectra.}
\label{tab::spec}
{ \footnotesize\begin{tabular}{c|c}
\hline
Method  &  Loss  (SE) $\times 10^{-5}$\\ \hline
   {NW} &  6.13 (1.47) \\ 
   {kNN} &   6.37 (1.52)\\
   \hline
       {Series}-poly1 & 5.13 (0.47)\\ 
    {Series}-polyBest (q=5) & 3.22 (0.32) \\ 
     \color{red}   {Series}-radial&  \color{red} 2.77 (0.33)\\   \hline
       {KRR}-poly1&  5.01 (0.49) \\ 
  {KRR}-polyBest (q=6)&  3.05 (0.33) \\ 
   \color{red}   {KRR}-radial&  \color{red}  2.84 (0.33)  \\\hline  
     \end{tabular}}
\qquad \qquad
{ \footnotesize \begin{tabular}{c|c}
\hline
Method  &  Loss (SE) \\ \hline
 {locOLS}& --  \\ \hline
  {locRR}& --  \\ \hline
  {locEN} & --  \\ \hline
  {locPLS} & --  \\ \hline
  {locPCR} & --  \\ \hline
 {NEDE} & --  \\ \hline
  {NALEDE}& --  \\ \hline
  {NEDEP}& -- \\  \hline
  {NALEDEP}& --  \\  \hline
   {MALLER}&  3.11    (0.38)  \\  \hline
\end{tabular}}
\end{center} 
\end{table}
\begin{figure}[H]
  \centering
   \includegraphics[width=3.5in,angle=-0]{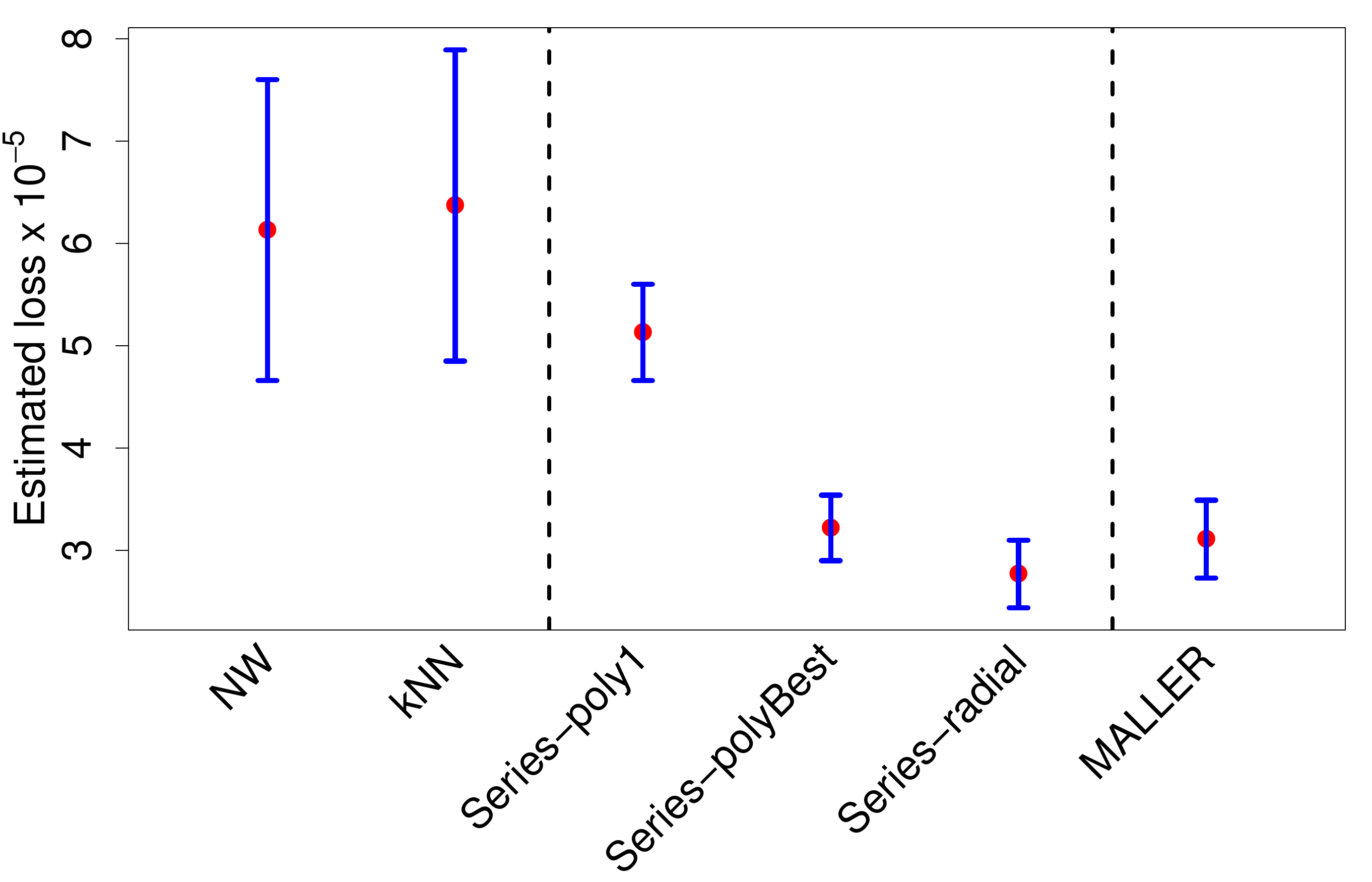}
  \caption{ \footnotesize Estimated loss of different estimators for redshift prediction; see Table~\ref{tab::spec}.  Bars represent standard errors. For visibility, we have divided the estimators into 3 groups: classical NW kernel and kNN smoothers ({\em left}),  Series/RKHS-type estimators  ({\em center}), and the manifold regression estimator \emph{MALLER} ({\em right}). Bars represent standard errors.}
  \label{fig::LossesSpec}
\end{figure}





\subsection{Scalability}

\subsubsection*{Increasing Dimension}
In terms of computational speed, the spectral series estimator has a clear competitive edge in high dimensions relative local regression procedures and a least-squares (LS) implementation of Eq.~\ref{eq::regularization}  that does not take advantage of orthogonal bases (see, e.g.,~\cite[p.~215]{Belkin:Niyogi:04} for a LS implementation of SSL learning on manifolds). We illustrate the differences with a one-dimensional manifold embedded in $d$ dimensions. 
Let
 $Y|\x \sim N(\theta(\x),0.5),$
where the points $\x=(x_1,\ldots,x_d)$ lie on a unit circle in $\mathbb{R}^d$,
  and $\theta(\x)$ is the angle corresponding to the position of $\x$. For simplicity, we simulate data uniformly on  the circle; i.e., we let $\theta(\x) \sim Unif(0,2\pi)$.

Figure \ref{fig::circle} summarizes the results. In terms of estimated loss (left panel),  \textit{Series} performs better than \textit{MALLER}, and it has a statistical performance similar to the least-squares implementation of kernel ridge regression (\textit{KRR-LS}). As predicted by the theory, the loss of \textit{Series} does not depend on the ambient dimension $d$. 
Moreover, the computational time of \textit{Series} is nearly constant as a function of the dimension $d$ (right panel). \textit{KRR-LS} is slower than \textit{Series},\footnote{Cross-validation of {\em Series} is fast due to the {\em orthogonality} of the basis. If we compute  the expansion coefficients $\hat{\beta}_{\varepsilon,j}$ (Eq.~\ref{eq::betahat}) for all $j \leq J_{\rm{max}}$, then we do not need to recompute these coefficients for other models with the same kernel and $J \leq J_{\rm{max}}$ components in the series expansion. The least squares implementation of Eq.~\ref{eq::regularization}, on the other hand, requires recomputing the expansion coefficients for each choice of the smoothing parameter $J$.} and \textit{MALLER} becomes computationally intractable as $d$ increases.  For $d=2500$ and $n=2000$, each  fit with \textit{MALLER} takes an average of 354 seconds (6 minutes) on an Intel i7-4800MQ CPU 2.70GHz processor, compared to 72 seconds for \textit{Series}. 

 \begin{figure}[H]
 	\begin{center}
 		\subfloat{    \includegraphics[width=2.45in,angle=-0]{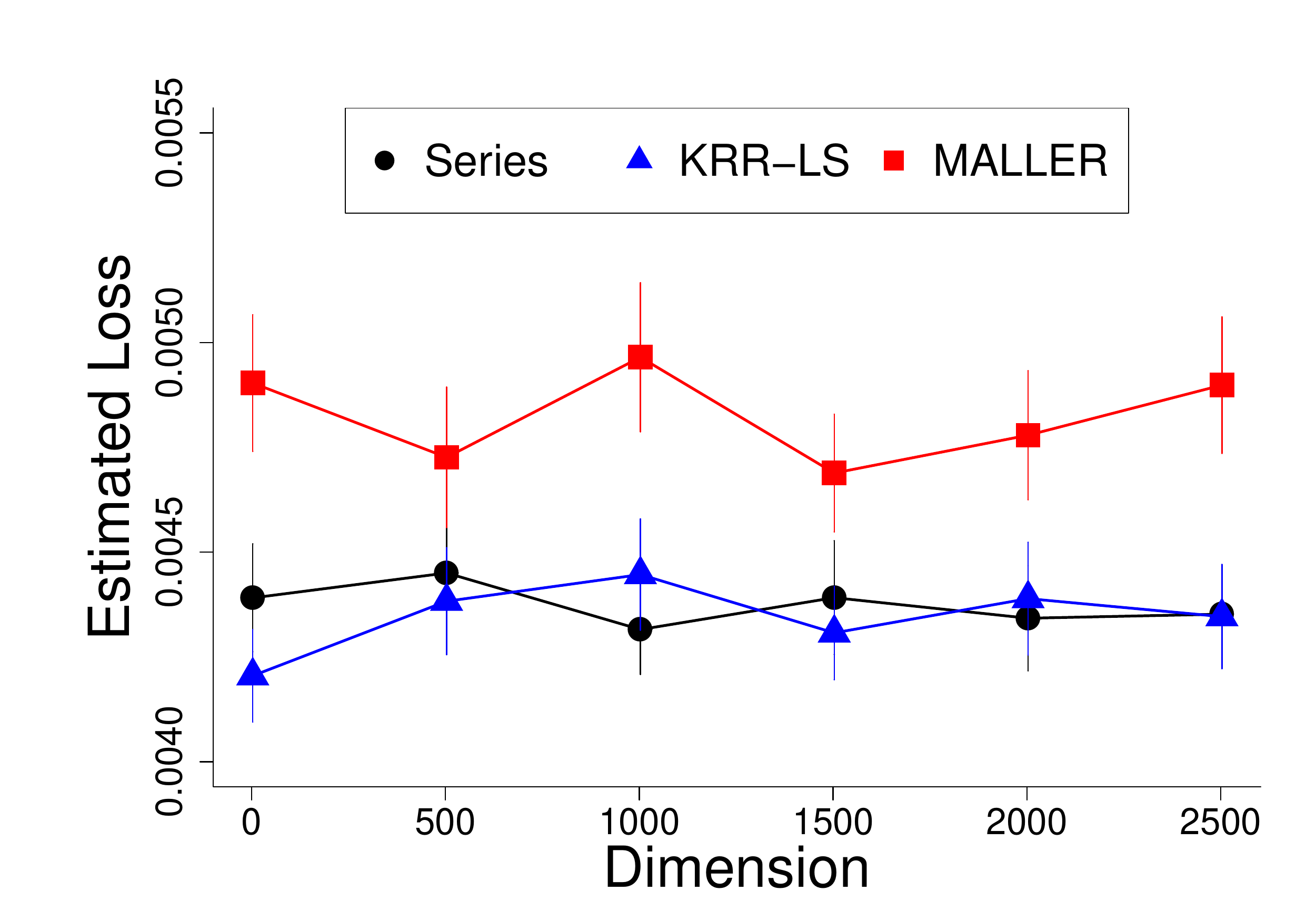} }
 		\subfloat{    \includegraphics[width=2.45in,angle=-0,page=2]{circleComparisons.pdf} }
 	\end{center}
 	\caption{ \footnotesize Increasing the dimension (number of variables) $d$ for a circle embedded in $\mathbb{R}^d$. Estimated loss ({\em left}) and computational time  ({\em right}) as a function of the dimension  $d$ for different regression estimators.}
 	\label{fig::circle}
 \end{figure}
 


\subsubsection*{Increasing Sample Size}
Here we revisit the redshift prediction problem in Sec.~\ref{eq:SDSS_spectra} using galaxy spectra from SDSS DR 12.\footnote{\url{http://www.sdss.org/dr12/}} We increase the size of the training set for a fixed number of $1000$ validation spectra and $1000$ test spectra. Fig.~\ref{fig::RSVD} indicates massive 
 payoffs in implementing Randomized SVD ({\em Series RSVD}) for large data sets;  see discussion in Sec.~\ref{sec::scalability}.    
Even without parallelization, we are able to cut down the computational time with a factor of 15 (right panel) with almost no decrease in statistical performance (left panel). The run time for {\em SVD} and {\em KRR-LS} when the sample size n=11200 (and the dimension $d=3431$) is about  5 hours on an Intel i7-4800MQ CPU 2.70GHz processor. With Randomized SVD, the same regression takes about 20 minutes. 
 \begin{figure}[H]
  \begin{center}
  \subfloat{    \includegraphics[width=2.45in,angle=-0]{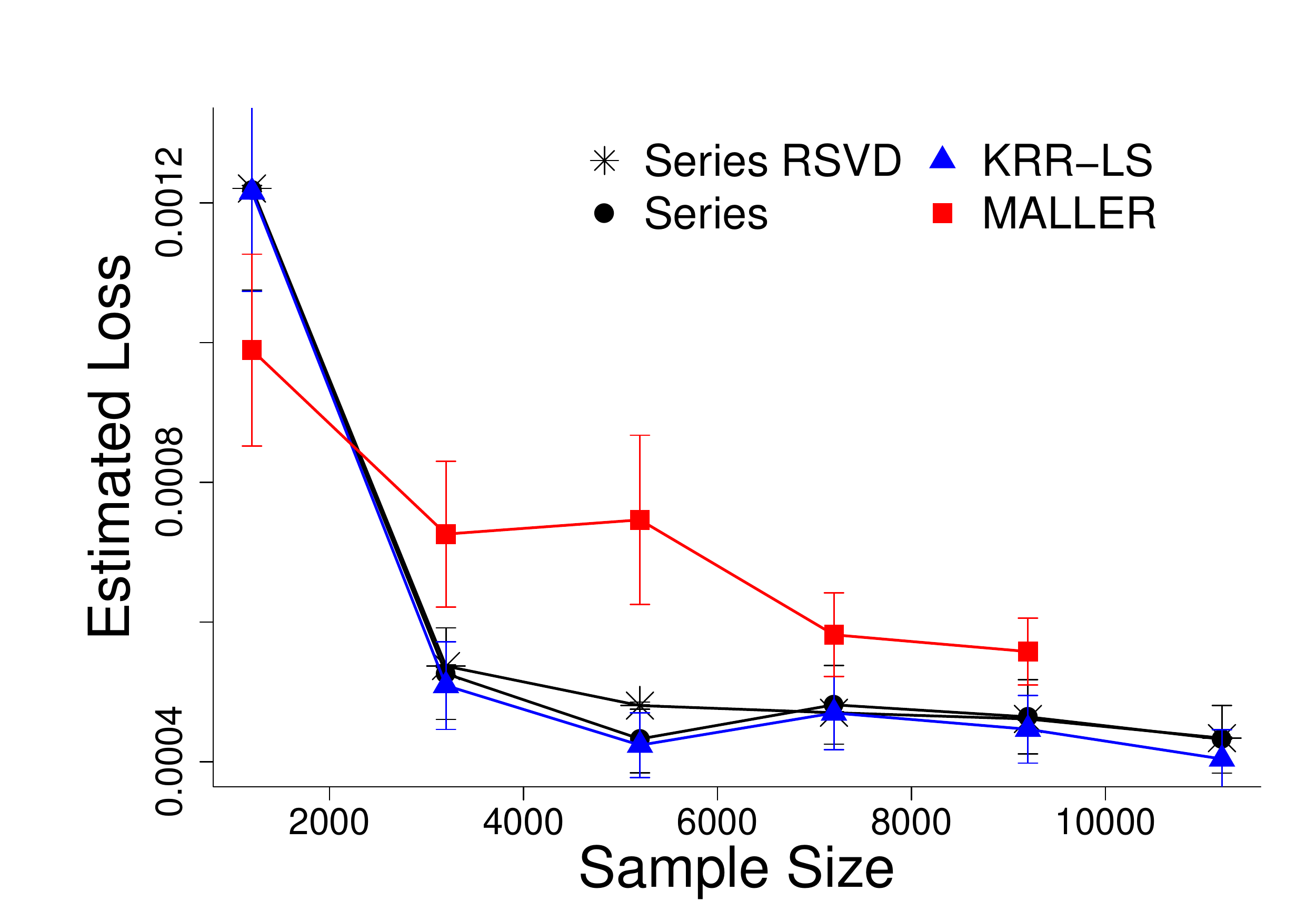} }
\subfloat{    \includegraphics[width=2.45in,angle=-0,page=2]{randomizedSVDSpectraNew.pdf} }
  \end{center}
  \caption{ \footnotesize Increasing the size of the training set for redshift prediction using SDSS galaxy spectra. Estimated loss ({\em left}) and computational time  ({\em right}) as a function of the sample size for different regression estimators. Note that Randomized SVD ({\em Series RSVD}) dramatically reduces the computational time, see right, for large sample sizes with little impact on the statistical performance, see left.}
  \label{fig::RSVD}
\end{figure}

\comment{The reviewer asked us to report more details (compared with other competing methods) on the computational shortcuts in Sec 2.6 on Scalability. Maybe we can illustrate the time savings with randomized SVD (and memory savings by thresholding) for the SDSS spectra? We could compare SVD, randomized SVD and, if possible, include some data points for MALLER and KRR-LSS.}

  \comment{Using higher-order polynomials (e.g., \emph{series-poly2}), however, does {\em not} improve the results further. The reason is the high dimension of the problem. In $\mathbb{R}^p$, the kernel $k(x,y)=(\langle x,y\rangle + 1)^q$ has a total of $M = {p+q \choose d}$ eigenfunctions that span the space of polynomials of degree $q$. For $p=3501$ and $q=1$, we have 3506 eigenfunctions already. Adding more eigenfunctions seem to only increase the variance of the series estimator, even when the functions are chosen nonlinearly according to decreasing values of $|\hat{\beta}_j|^2$, or when penalized by  $|\hat{\beta}_j|^2/\lambda_j$ as in \emph{KRR}.}

\section{Discussion}~\label{sec::conclusions}
Our spectral series method can handle complex high-dimensional data objects in many settings where traditional nonparametric methods either perform poorly or are computationally intractable. The series method offers a compression of the data in terms of Fourier coefficients; it is computationally efficient (with regards to the dimension and size of the sample), and it returns orthogonal basis functions that {\em adapt} to low-dimensional structure in the data distribution. As a result, there is no need for cumbersome tensor products in high dimensions. 

Our work shows that for a Gaussian kernel with a flexible bandwidth, the computed eigenfunctions form a Fourier-like orthogonal basis for expressing smoothness relative to the underlying data distribution. More precisely, if a function is smooth with respect to the data distribution, then it is sparse in the eigenbasis with respect to the $L^2$ norm, and vice versa (Theorem~\ref{prop::errordecay}). Indeed, in the limit of the sample size $n \rightarrow \infty$, spectral series with a Gaussian kernel can be seen as a generalization of Fourier series to high dimensions and sparse structure (Sec.~\ref{sec::kernel}).

The two main theorems \ref{theorem::loss_fixedkernel} and \ref{theorem::loss_epskernel} provide theoretical bounds on the loss of the final regression estimator for a standard fixed RKHS setting as well as a setting where the kernel bandwidth varies with the sample size $n$. 
We show that spectral series regression with a Gaussian kernel is adaptive to {\em intrinsic} dimension when the bandwidth $\varepsilon_n \rightarrow 0$ (Corollary \ref{corollary:manifold_example}). In the case of a submanifold with dimension $r$ embedded in $\mathbb{R}^d$, the convergence rate of the estimator depends on the manifold dimension $r$ rather than the ambient dimension $d$. The adaption occurs automatically and does not involve manifold estimation. 
 Unlike~\cite{Bickel:Li:2007}, there is also no need to estimate the dimension of the manifold. We have found 
 that unless the goal is manifold estimation, there is little advantage in using manifold and local linear regression methods. 
 Such methods quickly become computationally intractable in high dimensions without a prior dimension reduction. On the other hand, the computational speed of spectral series does {\em not} depend on the ambient or intrinsic dimension of the data.   Moreover, it is unclear how manifold-based methods behave in more complex settings where there is sparse structure (e.g., high-density regions and clusters) but no well-defined submanifold.  


 Because of the close connection between spectral series and SVMs, we expect that our new findings (regarding adaptiveness, choice of kernel and the bandwidth)  will apply to kernel-based regularized empirical risk minimizers as well.  Indeed, our empirical results (Tables \ref{tab::faces} and \ref{tab::spec}) confirm that the performance of KRR using a Gaussian kernel with a flexible bandwidth is similar to that of spectral series regression. This suggests that one can exploit the advantages of spectral series in terms of interpretation, visualization, and analysis
 without any real down-sides. In the process of analyzing the performance of the spectral series estimator, we  shed light on the empirical success of SVMs  for sparse data, 
and we unify ideas from Fourier analysis, kernel machine learning and spectral clustering.  

Future work includes deriving tighter bounds for the convergence rate of spectral series and kernel-based empirical regularizers.  We believe that our estimated rates are on the conservative side as our derivations assume that the eigenvectors need to be accurately estimated. Empirical experiments, however, indicate that spectral series with approximate eigenvectors already outperform the $k$-nearest neighbor estimator which is {\em minimax optimal} with respect to local intrinsic dimension~\citep{kpotufe2011k}.   In a separate paper, we will discuss extensions of spectral series to estimating other unknown functions 
 (e.g., conditional densities, density ratios and likelihoods) for high-dimensional complex data and distributions.  Another interesting research question is whether one can further improve the performance of spectral series approaches by adaptive basis selection and nonlinear estimators that threshold the series expansion coefficients $|\hat \beta_{j}|$ as in wavelet thresholding~\cite{Mallat:2009}.
 


In the online supplementary materials, we include sample {\tt R} code for the spectral series estimator. This code has however not been optimized for speed, as we will leave the large-scale deployment on parallel platforms to future work.

\section{Acknowledgments}
We are grateful to Ronald R. Coifman and Larry Wasserman for the original discussions that led to this work. We would also like to thank Peter Freeman, Cosma Shalizi and Ryan Tibshirani for insightful comments on the manuscript.  This work was partially supported by the Estella Loomis McCandless Professorship, \emph{Conselho Nacional de Desenvolvimento Cient\'ifico e Tecnol\'ogico} (grant 200959/2010-7),   \emph{Funda\c{c}\~ao de Amparo \`a Pesquisa do Estado de S\~ao Paulo} (2014/25302-2), and NSF DMS-1520786.

\bibliographystyle{chicago}
{ \footnotesize 
\bibliography{paper,ABCDE}
}

\appendix

\section{Proofs for Bounds on the Regression Estimator}


We start by stating some useful lemmas.

\begin{lemma}  \label{lemma::LaplaceBeltrami}
\cite[Proposition 3]{Coifman:Lafon:06} For $f \in C^3(\X)$ and $ x \in \X \setminus \partial \X$,
$$ -\lim_{\varepsilon \rightarrow 0} G_{\varepsilon}^{*} =  \triangle.$$
If  $\mathcal{X}$ is a compact $C^\infty$ submanifold of $\mathbb{R}^d$, then $\triangle$ is the psd Laplace-Beltrami operator of  $\mathcal{X}$ defined by $\triangle f(x)= - \sum_{j=1}^{r}\frac{\partial^2f}{\partial s_j^2}(x)$, where $(s_1,\ldots,s_r)$ are the normal coordinates of the tangent plane at $x$.
\end{lemma}

\begin{lemma}  \label{lemma::s_eps}$\forall x \in \X$, 
$$ \frac{a}{b} \leq s_{\varepsilon}(x) \leq \frac{b}{a}$$ 
\end{lemma}

\begin{proof} $\forall  x \in \X$,  
$$\frac{  \inf_{x \in \X} p_\varepsilon(x) }{  \sup_{x \in \X} p_\varepsilon(x) } \leq s_{\varepsilon}(x) \leq \frac{ \sup_{x \in \X}p_\varepsilon(x) }{ \inf_{x \in \X}p_\varepsilon(x) },$$
where $a  \int_\X \! k_\varepsilon(x,y) dy \leq p_\varepsilon(x)  \leq b \int_\X \! k_\varepsilon(x,y) dy$. 
\end{proof}

\begin{lemma}\label{lemma::Lbias_conversion}
For $f  \in   L^2(\X,P)$,
$$L_{\rm{bias}} \leq \frac{b}{a}  \sum_{j>J} |\beta_{\varepsilon,j}|^2.$$ 
\end{lemma}
\begin{proof}
From the orthogonality property of the basis functions $\psi_j$, we have that 
$$\int_{\X}  | f(x) - f_{\varepsilon,J} (x) |^2 dS_{\varepsilon}(x) =  \sum_{j>J} |\beta_{\varepsilon,j}|^2.$$
The result follows from Lemma~\ref{lemma::s_eps}.
\end{proof}

\begin{lemma}  \label{lemma::phi_error}
Under the same assumptions as in Proposition~\ref{prop::eigenvec_error}, it holds that
$$\|\varphi_{\varepsilon,j} - \hat{\varphi}_{\varepsilon,j}\|_ {L^2(\mathcal{X},P)} = O_P\left(\frac{\gamma_n}{\delta_{\varepsilon,j}}\right),$$
where  $\gamma_n = \sqrt{\frac{ \log(1/\varepsilon_n)}{n \varepsilon_n^{d/2}}}$ and  $\delta_{\varepsilon,j}=\lambda_{\varepsilon,j}-\lambda_{\varepsilon,j+1}.$
\end{lemma}
\begin{proof} 
From \cite{Gine:Guillou:02}, $\sup_x  | \hat{p}_\varepsilon(x) - p_\varepsilon(x)| = O_P(\gamma_n)$. Hence, 
$$
\sup_x  | \hat{s}_\varepsilon(x) - s_\varepsilon(x)| = O_P(\gamma_n).  
 $$
By using Proposition~\ref{prop::eigenvec_error}, we conclude that
$$ \int_\X |\hat \psi_{\varepsilon,j}(x)|^2 dP(x) \leq  2\int_\X |\hat \psi_{\varepsilon,j}(x)-\psi_{\varepsilon,j}(x)|^2 dP(x)  + 2 \int_\X | \psi_{\varepsilon,j}(x)|^2 dP(x) = O_P\left(\frac{\gamma_n^2}{\delta_{\varepsilon,j}^2}\right) + C,$$
where $C$ is a constant. Write
 \begin{eqnarray*} | \varphi_{\varepsilon,j}(x) - \widehat{\varphi}_{\varepsilon,j}(x) |^2 &=&  | \psi_{\varepsilon,j}(x) s_{\varepsilon}(x) - \widehat{\psi}_{\varepsilon,j}(x) \widehat s_{\varepsilon}(x) | ^2 \\
  &\leq&   2  |\psi_{\varepsilon,j}(x)-\hat \psi_{\varepsilon,j}(x)|^2 |s_{\varepsilon}(x)|^2 +  2 |s_{\varepsilon}(x)-\widehat s_{\varepsilon}(x)|^2  |\hat \psi_{\varepsilon,j}(x)|^2.
 \end{eqnarray*}
Hence,
\begin{eqnarray*}\|\varphi_{\varepsilon,j} - \hat{\varphi}_{\varepsilon,j}\|_ {L^2(\mathcal{X},P)}^2 &\leq& 2 \sup_x |s_{\varepsilon}(x)|^2 \; \| \hat\psi_{\varepsilon,j}-\psi_{\varepsilon,j}\|_{L^2(\mathcal{X},P)}^2 +  2 \sup_x |\hat s_{\varepsilon}(x)-s_{\varepsilon}(x)|^2 \; \|\hat{\psi}_{\varepsilon,j}\|_ {L^2(\mathcal{X},P)}^2\\
&=& O_P\left(\frac{\gamma_n^2}{\delta_{\varepsilon,j}^2}\right)  + O_P(\gamma_n^2) \left(O_P\left(\frac{\gamma_n^2}{\delta_{\varepsilon,j}^2}\right) + C\right)\\ 
&=& O_P\left(\frac{\gamma_n^2}{\delta_{\varepsilon,j}^2}\right) .
\end{eqnarray*}
\end{proof}

\begin{lemma}  \label{lemma::Y_eigdev}
$\forall \  0 \leq j \leq J$,  it holds that
$$\left| \frac{1}{n}\sum_{i=1}^n  Y_i(\hat \varphi_{\varepsilon,j}(X_i)- \varphi_{\varepsilon,j}(X_i)) - \int_\X f(x)(\hat \varphi_{\varepsilon,j}(x)- \varphi_{\varepsilon,j}(x))dP(x) \right| =   O_P\left(\frac{1}{\sqrt{n}}\right).$$
\end{lemma}
\begin{proof} Let $S=\frac{1}{n}\sum_{i=1}^n  Y_i(\hat \varphi_{\varepsilon,j}(X_i)- \varphi_{\varepsilon,j}(X_i))$ and $I= \int_\X f(x)(\hat \varphi_{\varepsilon,j}(x)- \varphi_{\varepsilon,j}(x))dP(x)$. According to Chebyshev's inequality, for any $M>0$,
\begin{eqnarray*} \CP{|S-I| \geq M}{\tilde{X}_1,\ldots,  \tilde{X}_n} &\leq& \frac{\mathbb{V} \left( S-I |  \tilde{X}_1,\ldots,  \tilde{X}_n \right)}{M^2} 
\leq \frac{\CV{Y_1 (\hat\varphi_{\varepsilon,j}(X_1)- \varphi_{\varepsilon,j}(X_1))}{\tilde{X}_1,\ldots,  \tilde{X}_n}}{nM^2}.\end{eqnarray*}
Hence, for any $M>0$,
$$\mathbb{P} \left(|S-I| \geq M \right) \leq \frac{\mathbb{V} \left( Y_1 (\hat\varphi_{\varepsilon,j}(X_1)- \varphi_{\varepsilon,j}(X_1))\right)}{nM^2}
\leq  \frac{\sigma}{nM^2} \left(\E |\hat\varphi_{\varepsilon,j}(X_1)- \varphi_{\varepsilon,j}(X_1)|^2\right)^{1/2},$$
where we in the last inequality apply the Cauchy-Schwarz inequality.
Under assumption (A4), we conclude the result of the lemma.
\end{proof}

\subsection{Bias}

\noindent{\bf Proof of Proposition~\ref{prop::bias_RKHS}.}
Let $\tilde{\psi}_{\varepsilon,0},\tilde{\psi}_{\varepsilon,1},\ldots$ be the eigenfunctions of the symmetric operator $\tilde{A}_\varepsilon (f)(x) = \int_{\mathcal{X}}
\tilde{a}_\varepsilon(x,y) f(y) dP(y)$. It follows from Mercer's theorem that 
$$\mathcal{H}_{\varepsilon,M} = \left\{f = \sum_j  \gamma_{\varepsilon,j}\tilde{\psi}_{\varepsilon,j}   :  \sum_j \frac{|\gamma_{\varepsilon,j}|^2}{\lambda_{\varepsilon,j}} \leq M \right\}.$$
If $f \in \mathcal{H}_{\varepsilon,M}$, then we can bound the bias with respect to the eigenbasis $\tilde{\psi}$:
$$   \sum_{j>J} | \gamma_{\varepsilon,j}|^2 \leq  \lambda_{\varepsilon,J} \sum_{j>J}  \frac{|\gamma_{\varepsilon,j}|^2}{\lambda_{\varepsilon,j}} 
 \leq  \lambda_{\varepsilon,J} \sum_{j}  \frac{|\gamma_{\varepsilon,j}|^2}{\lambda_{\varepsilon,j}} \leq M \lambda_{\varepsilon,J}. $$
 By construction,
 $$\gamma_{\varepsilon,j} = \int_\X f(x) \tilde{\psi}_{\varepsilon,j}(x) dP(x)$$ 
 and
 $$\beta_{\varepsilon,j} = \int_\X f(x) \psi_{\varepsilon,j}(x) dS_\varepsilon(x) = \int_\X f(x) \tilde{\psi}_{\varepsilon,j}(x) \sqrt{s_\varepsilon(x)} dP(x) \leq  \sqrt{\frac{b}{a}} |\gamma_{\varepsilon,j}|.$$ 
 Thus, $\displaystyle   \sum_{j>J} | \beta_{\varepsilon,j}|^2 \leq  \frac{b}{a} \sum_{j>J} | \gamma_{\varepsilon,j}|^2  =  O(M \lambda_{\varepsilon,J}).$ 
 The result follows from Lemma~\ref{lemma::Lbias_conversion}.~$\Box$\\

\noindent{\bf Proof of Proposition~\ref{proposition::bias_diffop}.}
Note that $\mathcal{J}_{\varepsilon}(f) = \sum_{j}  \nu_{\varepsilon,j}^2 |\beta_{\varepsilon,j}| ^2$. Hence,
$$ \frac{ \mathcal{J}_{\varepsilon}(f)}{\nu^2_{\varepsilon,J+1}} =
 \sum_{j } \frac{\nu_{\varepsilon,j}^2}{ \nu_{\varepsilon,J+1}^2}  |\beta_{\varepsilon,j}|^2 \geq  \sum_{j>J} \frac{\nu_{\varepsilon,j}^2}{ \nu_{\varepsilon,J+1}^2}  |\beta_{\varepsilon,j}|^2 \geq \sum_{j>J} |\beta_{\varepsilon,j}|^2 = \int_{\X}  | f(x) - f_{\varepsilon,J} (x) |^2 dS_{\varepsilon}(x).$$ The last result follows from Lemma~\ref{lemma::Lbias_conversion}.~$\Box$\\

\noindent{\bf Proof of Lemma~\ref{lemma:conv_smoothfun}.}
By Green's first identity
$$ \int_\X f \;\nabla^2 \! f dS(x) + \int_\X \nabla f \cdot \nabla f dS(x) = \oint_{\partial \X} f (n \cdot \nabla f) dS(x)=0,$$
where $n$ is the normal direction to the boundary $\partial \X$, and the last surface integral vanishes due to the Neumann boundary condition.    It follows from  Lemma~\ref{lemma::LaplaceBeltrami} that 
$$ \lim_{\varepsilon \rightarrow 0} \mathcal{J}^{*}_{\varepsilon}(f)  =  - \lim_{\varepsilon \rightarrow 0} \int_\X f(x) G_{\varepsilon}^{*} f(x) dS_\varepsilon(x)= \int_\X f(x) \triangle \! f(x) dS(x) =
\int_{\X} \|\nabla f(x)\|^2 dS(x).\, \Box$$\\

 \noindent{\bf Proof of Theorem~\ref{prop::errordecay}.}
We have that 
$$ c^2 \geq \int_{\X} \|\nabla f(x)\|^2 dS(x) = \int_\X f(x) \triangle \! f(x) dS(x) = \sum_j \nu_j^2 \beta_j^2,$$
where $\nu_j^2=O(j^{2s})$. Hence, $\displaystyle f \in  W_\mathcal{B}(s,c)$ and by Theorem 9.1 in~\cite{Mallat:2009}, $\|f - f_J \|^2 = o(J^{-2s}).$~$\Box$\\

\subsection{Variance}~\label{appendix::variance}

 Let $\mathcal{H}$ be an auxiliary RKHS of smooth functions; 
  we use the term ``auxiliary'' to denote that the space only enters the intermediate derivations and plays no role in the error analysis of the algorithm itself.  We define the two integral operators $A_\mathcal{H}, \widehat{A}_\mathcal{H}: \mathcal{H} \rightarrow \mathcal{H}$ where 
 \begin{eqnarray*}
 A_{\mathcal{H}} f(x) &=&   \frac{\int k_\varepsilon(x ,y) \langle f, K(\cdot,y) \rangle_{\mathcal{H}}  dP(y)}{\int k_\varepsilon(x,y)dP(y)} = \int a_\varepsilon(x,y) \langle f, K(\cdot,y) \rangle_{\mathcal{H}} \; dP(y) \\
 \widehat{A}_{\mathcal{H}} f(x) &=& \frac{\sum_{i=1}^n k_\varepsilon(x,X_i) \langle f, K(\cdot,X_i) \rangle_{\mathcal{H}}}{\sum_{i=1}^n k_\varepsilon(x,X_i)}=
\int \hat{a}_\varepsilon(x,y) \langle f, K(\cdot,y) \rangle_{\mathcal{H}} \; d\hat{P}_n(y) ,
\end{eqnarray*}
and  $K$ is the reproducing kernel of  $\mathcal{H}$. Define the operator norm $\|A\|_\mathcal{H}= \sup_{f\in {\cal H}} 
\|Af\|_\mathcal{H}/\|f\|_\mathcal{H}$ where  $\|f\|_\mathcal{H}^2=\langle f,f\rangle_{\cal H}$. Now suppose the weight function $k_\varepsilon$ is sufficiently smooth with respect to $\mathcal{H}$ (Assumption 1 in~\cite{Rosasco:EtAl:2010}); this condition is for example satisfied by a Gaussian kernel on a compact support $\mathcal{X}$.
By Propositions 13.3 and 14.3 in~\cite{Rosasco:EtAl:2010}, we can then relate the functions $\psi_{\varepsilon,j}$ and $\widehat \psi_{\varepsilon,j}$, respectively, to the eigenfunctions $ u_{\varepsilon,j}$ and $\widehat {u}_{\varepsilon,j}$ of $A_\mathcal{H}$ and $\widehat{A}_\mathcal{H}$. We have that  
\begin{equation}
 \| \psi_{\varepsilon,j} - \widehat \psi_{\varepsilon,j}\|_{L^2(\mathcal{X},P)} = C_1 \|  u_{\varepsilon,j}-\widehat {u}_{\varepsilon,j}\|_{L^2(\mathcal{X},P)} \leq C_2 \|  u_{\varepsilon,j}-\widehat {u}_{\varepsilon,j}\|_\mathcal{H}
 \label{eq::eigrel_auxRKHS}
 \end{equation}
for some constants $C_1$ and $C_2$.  According to Theorem 6 in~\cite{Rosasco:EtAl:2008} for eigenprojections of positive compact operators, it holds that
 \begin{equation}
 \|  u_{\varepsilon,j}-\widehat {u}_{\varepsilon,j}\|_\mathcal{H} \leq \frac{\| A_\mathcal{H} - \widehat{A}_\mathcal{H} \|_\mathcal{H}}{ \delta_{\varepsilon,j}},  \label{eq:pert_eigf}
  \end{equation}
  where $ \delta_{\varepsilon,j}$ is proportional to the eigengap $\lambda_{\varepsilon,j} - \lambda_{\varepsilon,j+1}$. As a result, we can bound the difference $\|\psi_{\varepsilon,j} - \widehat{\psi}_{\varepsilon,j}\|_{L^2(\mathcal{X},P)}$ by controlling  the deviation $\|A_\mathcal{H}-\widehat{A}_\mathcal{H}\|_{\cal H}$.

We choose the auxiliary RKHS $\mathcal{H}$ to be a Sobolev space with a sufficiently high degree of smoothness 
(see below for details). Let $\mathcal{H}^s$ denote the Sobolev space of order $s$ with vanishing gradients at the boundary; that is, let
$$ \mathcal{H}^s = \{f \in L^2(\mathcal{X}) \; |  \; D^{\alpha}f \in L^2(\mathcal{X})   \, \,  \forall |\alpha|\leq s, \; D^{\alpha} f|_{\partial \mathcal{X}} =0  \,  \,  \forall |\alpha|=1\},$$
where $D^{\alpha}f$ is the weak partial derivative of $f$ with respect to the multi-index $\alpha$, and $L^2(\mathcal{X})$ is the space of square integrable functions with respect to the Lebesgue measure. 
Let $C_b^3(\mathcal{X})$ be the set of uniformly bounded, three times differentiable functions with uniformly bounded derivatives whose gradients vanish at the boundary. 
Now 
 consider $\mathcal{H} \subset \mathcal{H}^s$ 
 and choose $s$ large enough so that $D^{\alpha}f  \in C_b^3(\mathcal{X})$ for all $ f \in \mathcal{H}$ and $|\alpha|=s$. 
Under assumptions  (A1)-(A4), we derive the following result:
\begin{lemma}\label{lemma::Aest}
Let $\varepsilon_n\to 0$ and $n \varepsilon_n^{d/2}/
\log(1/\varepsilon_n) \to \infty$.
Then 
$\|A_\mathcal{H} - \widehat{A}_\mathcal{H}\|_{\mathcal{H}} = O_P(\gamma_n),$
where 
$
\gamma_n =\sqrt{\frac{ \log(1/\varepsilon_n)}{n \varepsilon_n^{d/2}}}.$
\end{lemma}
\begin{proof}
Uniformly, for all $f \in C_b^3(\mathcal{X})$, and all $x$ in the support of $P$,
$$
|A_\varepsilon f(x) - \hat{A}_\varepsilon f(x)| \leq
|A_\varepsilon f(x) - \tilde{A}_\varepsilon f(x)| +
|\tilde{A}_\varepsilon f(x) - \hat{A}_\varepsilon f(x)|
$$
where
$\tilde{A}_\varepsilon f(x) = \int \hat{a}_\varepsilon(x,y)f(y) dP(y).$
\relax From \cite{Gine:Guillou:02},
$$
\sup_x \frac{| \hat{p}_\varepsilon(x) - p_\varepsilon(x)|}
            {| \hat{p}_\varepsilon(x)p_\varepsilon(x)|} =
  O_P(\gamma_n).  
$$
Hence,
\begin{eqnarray*}
|A_\varepsilon f(x) - \tilde{A}_\varepsilon f(x)| & \leq &
\frac{| \hat{p}_\varepsilon(x) - p_\varepsilon(x)|}
{| \hat{p}_\varepsilon(x)p_\varepsilon(x)|}\int |f(y)| k_
\varepsilon(x,y)dP(y) \\
&= &  O_P(\gamma_n)   
\int |f(y)| k_\varepsilon(x,y)dP(y)\\
& = &
O_P (\gamma_n). 
\end{eqnarray*}
Next, we bound
$\tilde{A}_\varepsilon f(x) - \hat{A}_\varepsilon f(x)$.
We have
\begin{eqnarray*}
\tilde{A}_\varepsilon f(x) - \hat{A}_\varepsilon f(x) & = &
\int f(y) \hat{a}_\varepsilon(x,y) (d\hat{P}_n(y) - dP(y))\\
&=& \frac{1}{p(x)+o_P(1)}\int f(y) k_\varepsilon(x,y)(d\hat{P}_n(y) - 
dP(y)).
\end{eqnarray*}
Now, expand
$f(y) = f(x) + r_n(y)$
where $r_n(y) = (y-x)^T\nabla f(u_y)$ and
$u_y$ is between $y$ and $x$.
So,
$$
\int f(y) k_\varepsilon(x,y)(d\hat{P}_n(y) - dP(y)) =
f(x) \int  k_\varepsilon(x,y)(d\hat{P}_n(y) - dP(y))  +
 \int r_n(y) k_\varepsilon(x,y)(d\hat{P}_n(y) - dP(y)) .
$$
By an application of Talagrand's inequality
to each term, as in
Theorem 5.1 of   \cite{Gine:Koltchinskii:2006},   
we have
$$
\int f(y) k_\varepsilon(x,y)(d\hat{P}_n(y) - dP(y)) =
O_P(\gamma_n).
$$
Thus,
$
\sup_{f\in C_b^3(\mathcal{X})}\| \hat{A}_\varepsilon f - A_\varepsilon f \|_\infty = 
O_P(\gamma_n).$

The Sobolev space $\mathcal{H}$ is a Hilbert space with respect to the scalar product 
$$\langle f,g\rangle_{\mathcal{H}}= \langle f,g\rangle_{L^2(\mathcal{X})} +     \sum_{|\alpha|=s} \langle D^{\alpha}f,D^{\alpha}g\rangle_{L^2(\mathcal{X})}.$$
We have that
\begin{eqnarray*}
  \sup_{f\in {\cal H}:\|f\|_\mathcal{H}=1} \| \hat{A}_\varepsilon f - A_\varepsilon f\|_\mathcal{H}^2 &\leq&  \sup_{f\in {\cal H}} \sum_{|\alpha| \leq s} \| D^{\alpha}(\hat{A}_\varepsilon f - A_\varepsilon f)\|_{L^2(\mathcal{X})}^2 = \sum_{|\alpha| \leq s} \sup_{f\in {\cal H}} \| \hat{A}_\varepsilon D^{\alpha}f - A_\varepsilon D^{\alpha}f\|_{L^2(\mathcal{X})}^2 \\
  &\leq& \sum_{|\alpha| \leq s} \sup_{f \in C_b^3(\mathcal{X})} \| \hat{A}_\varepsilon f - A_\varepsilon f\|_{L^2(\mathcal{X})}^2 \leq C \sup_{f\in C_b^3(\mathcal{X})}\| \hat{A}_\varepsilon f - A_\varepsilon f \|_\infty^2.\
  \end{eqnarray*}
  for some constant $C$.  Hence, 
  \begin{eqnarray*}
\sup_{f\in {\cal H}} \frac{\| \hat{A}_\varepsilon f - A_\varepsilon f \|_\mathcal{H}}{\|f\|_\mathcal{H}} 
=
\sup_{f\in {\cal H}, \|f\|_\mathcal{H}=1} \| \hat{A}_\varepsilon f - A_\varepsilon f \|_\mathcal{H} \leq C'  \sup_{f\in C_b^3(\mathcal{X})}\| \hat{A}_\varepsilon f - A_\varepsilon f \|_\infty =
O_P(\gamma_n). 
\end{eqnarray*}
\end{proof}

\noindent{\bf Proof of Proposition~\ref{prop::eigenvec_error}.}
  From Eqs.~\ref{eq::eigrel_auxRKHS}-\ref{eq:pert_eigf}, we have that
\begin{equation*}
   \| \psi_{\varepsilon,j} - \widehat \psi_{\varepsilon,j}\|_{L^2(\mathcal{X},P)}  \leq  C \, \frac{\| A_\mathcal{H} - \widehat{A}_\mathcal{H} \|_\mathcal{H}}{\lambda_{\varepsilon,j} - \lambda_{\varepsilon,j+1}}
  \end{equation*}
  for some constant $C$ that does not depend on $n$.
  The result follows from Lemma~\ref{lemma::Aest}.~$\Box$\\
  
   \comment{Furthermore, for all $x \in \cal{X}$, 
    $$ | \varphi_{\varepsilon,j}(x) - \widehat{\varphi}_{\varepsilon,j}(x) | =  | \psi_{\varepsilon,j}(x) s_{\varepsilon}(x) - \widehat{\psi}_{\varepsilon,j}(x) \widehat s_{\varepsilon}(x) | \stackrel{\fbox{??}}{\leq}
  B  | \psi_{\varepsilon,j}(x)  - \widehat{\psi}_{\varepsilon,j}(x)  |  $$
  for some constant $B$.
  The second result follows.}

\begin{lemma}  \label{lemma::betahat}
$\forall 0 \leq j \leq J$, 
$$ |\hat{\beta}_{\varepsilon,j}-\beta_{\varepsilon,j}|^2 = O_P\left(\frac{1}{n}\right)  + O_P\left(\frac{\gamma_n^2}{\delta_{\varepsilon,j}^2}\right).$$
\end{lemma}

\begin{proof} 
Note that $\psi_{\varepsilon,j}(x) s_{\varepsilon}(x) = \varphi_{\varepsilon,j}(x)$ and 
\begin{eqnarray*}
\hat{\beta}_{\varepsilon,j} &=& \frac{1}{n} \sum_{i=1}^n Y_i \hat \psi_{\varepsilon,j}(X_i) \hat{s}_{\varepsilon}(X_i) \\
&=& \frac{1}{n} \sum_{i=1}^n Y_i \varphi_{\varepsilon,j}(X_i)  + \frac{1}{n} \sum _{i=1}^n Y_i  \left(\hat \varphi_{\varepsilon,j}(X_i) -    \varphi_{\varepsilon,j}(X_i)\right)   \\
 &=& \beta_{\varepsilon,j} + O_P\left(\frac{1}{\sqrt{n}}\right) +  \frac{1}{n} \sum_{i=1}^n Y_i   \left(\hat \varphi_{\varepsilon,j}(X_i)- \varphi_{\varepsilon,j}(X_i)\right).  
 \end{eqnarray*}
 Let $S= \frac{1}{n} \sum _{i=1}^n Y_i  \left(\hat \varphi_{\varepsilon,j}(X_i) -    \varphi_{\varepsilon,j}(X_i)\right)$ and  $I= \int_\X f(x)(\hat \varphi_{\varepsilon,j}(x)- \varphi_{\varepsilon,j}(x))dP(x)$. 
We conclude that
 \begin{eqnarray}
\frac{1}{2}|\hat{\beta}_{\varepsilon,j}-\beta_{\varepsilon,j}|^2 &\leq&   O_P\left(\frac{1}{n}\right) + |S-I|^2+  |I|^2 \nonumber\\
&\leq&   O_P\left(\frac{1}{n}\right)  +|S-I|^2 +  \left(\int_\X |f(x)|^2 dP(x) \right) \left(\int_{\X} | \varphi_{\varepsilon,j}(x)-\hat{\varphi}_{\varepsilon,j}(x)  |^2  dP(x) \right)\nonumber  \\ 
&=&   O_P\left(\frac{1}{n}\right) +  O_P\left(\frac{\gamma_n^2}{\delta_{\varepsilon,j}^2}\right),\nonumber
\end{eqnarray}  
   where the second inequality follows from the Cauchy-Schwarz inequality, and the last equality is due to Lemmas~\ref{lemma::Y_eigdev} and~\ref{lemma::phi_error}.    
\end{proof}
   
\noindent{\bf Proof of Proposition~\ref{prop::variance}.} 
Let $\tilde{f}_{\varepsilon,J}(x) = \sum_{j=0}^{J}  \beta_{\varepsilon,j}  \hat\psi_{\varepsilon,j}(x)$.
 Write
\begin{eqnarray*} |f_{\varepsilon,J}(X_i) - \hat{f}_{\varepsilon,J}(X_i)|^2 &=& | f_{\varepsilon,J}(X_i)  - \tilde{f}_{\varepsilon,J}(X_i) + \tilde{f}_{\varepsilon,J}(X_i)    - \hat{f}_{\varepsilon,J}(X_i)|^2 \\
&\leq& 2 | f_{\varepsilon,J}(X_i)  - \tilde{f}_{\varepsilon,J}(X_i) |^2  + 2 | \tilde{f}_{\varepsilon,J}(X_i)  - \hat f_{\varepsilon,J}(X_i) |^2
\end{eqnarray*}
We bound the contribution  to $L_{\rm{var}}$ from each of these two terms  separately:

By using Cauchy's inequality  and Proposition~\ref{prop::eigenvec_error}, we have that
\begin{eqnarray*}
&&\int_{\X} | f_{\varepsilon,J}(x)  - \tilde{f}_{\varepsilon,J}(x) |^2 dP(x)=
\int_{\X} \left|   \sum_{j=0}^{J}  \beta_{\varepsilon,j}  (\psi_{\varepsilon,j}(x) - \hat{\psi}_{\varepsilon,j}(x) ) \right |^2 dP(x) \\
 &\leq& \left(\sum_{j=0}^J |  \beta_{\varepsilon,j}|^2 \right) \cdot
 \sum_{j=0}^{J}  \left( \int_{\X}  |  \psi_{\varepsilon,j}(x)-\hat{\psi}_{\varepsilon,j}(x)  |^2  dP(x) \right) = J \ O_P\left(\frac{\gamma_n^2}{\Delta_{\varepsilon,J}^2}\right) .
 \end{eqnarray*}

By construction, it holds that  $\frac{1}{n}\sum_i \hat\psi_{\varepsilon,j}(\widetilde X_i) \hat\psi_{\varepsilon,\ell}(\tilde X_i) \hat{s}_\varepsilon(\tilde X_i)=\delta_{j,\ell}$. Furthermore, 
\begin{eqnarray*} \int_{\X} \hat \psi_{\varepsilon,j}(x)  \hat \psi_{\varepsilon,\ell}(x) d\hat{S}_\varepsilon(x) &=&  \frac{1}{n} \sum_i \hat\psi_{\varepsilon,j}(X_i) \hat\psi_{\varepsilon,\ell}(X_i) \hat{s}_\varepsilon(X_i) + O_P\left(\frac{1}{\sqrt{n}}\right)\\
&=&   \frac{1}{n} \sum_i \hat\psi_{\varepsilon,j}(\tilde X_i) \hat\psi_{\varepsilon,\ell}(\tilde X_i) \hat{s}_\varepsilon(\tilde X_i) + O_P\left(\frac{1}{\sqrt{n}}\right) \\
&=&\delta_{j,\ell} + O_P\left(\frac{1}{\sqrt{n}}\right).
\end{eqnarray*}
for a sample $X_1,\ldots,X_n$ drawn independently from  $\tilde X_1,\ldots,\tilde X_n$. Finally, from the orthogonality property of the $\hat\psi_{\varepsilon,j}$'s together with Lemmas~\ref{lemma::betahat} and~\ref{lemma::s_eps}, it follows that
 \begin{eqnarray*}
&& \int_{\X} | \tilde{f}_{\varepsilon,J}(x)  - \hat f_{\varepsilon,J}(x) |^2 dP(x) =
  \int_{\X}  \frac{1}{\hat{s}_\varepsilon(x)}\left|\sum_{j=0}^J ( \beta_{\varepsilon,j} -  \hat{\beta}_{\varepsilon,j}  )  \hat\psi_{\varepsilon,j}(x) \sqrt{\hat{s}_\varepsilon(x)}\right|^2 dP(x) \\
&=& \int_{\X}  \frac{1}{\hat{s}_\varepsilon(x)}
\left(\sum_{j=0}^J ( \beta_{\varepsilon,j} -  \hat{\beta}_{\varepsilon,j}  )^2  \hat\psi_{\varepsilon,j}^2(x) d\hat{S}_\varepsilon(x)\right)\\
&&+ \int_{\X}  \frac{1}{\hat{s}_\varepsilon(x)}
\left(\sum_{j=0}^J  \sum_{\ell=0, \ell \neq j}^J( \beta_{\varepsilon,j} -  \hat{\beta}_{\varepsilon,j}  ) ( \beta_{\varepsilon,\ell} -  \hat{\beta}_{\varepsilon,\ell}  )  \hat\psi_{\varepsilon,j}(x) \hat\psi_{\varepsilon,\ell}(x) d\hat{S}_\varepsilon(x) 
\right) \\
&\leq& \frac{b}{a}   \sum_{j=0}^J ( \beta_{\varepsilon,j} -  \hat{\beta}_{\varepsilon,j}  )^2 \left( \int_{\X} \hat\psi_{\varepsilon,j}^2(x) d\hat{S}_\varepsilon(x)\right)\\
&+&  \frac{b}{a}  \sum_{j=0}^J  \sum_{\ell=0, \ell \neq j}^J( \beta_{\varepsilon,j} -  \hat{\beta}_{\varepsilon,j}  ) ( \beta_{\varepsilon,\ell} -  \hat{\beta}_{\varepsilon,\ell}  ) \left( \int_{\X} \hat\psi_{\varepsilon,j}(x) \hat\psi_{\varepsilon,\ell}(x) d\hat{S}_\varepsilon(x) 
\right)\\
&=& \frac{b}{a}   \sum_{j=0}^J ( \beta_{\varepsilon,j} -  \hat{\beta}_{\varepsilon,j}  )^2 \left(1+O_P\left(\frac{1}{\sqrt{n}} \right)\right)
+  \frac{b}{a}  \sum_{j=0}^J  \sum_{\ell=0, \ell \neq j}^J( \beta_{\varepsilon,j} -  \hat{\beta}_{\varepsilon,j}  ) ( \beta_{\varepsilon,\ell} -  \hat{\beta}_{\varepsilon,\ell}  ) \; O_P\left(\frac{1}{\sqrt{n}}\right)\\
&=& J  \left( O_P \left( \frac{1}{n}\right)  + O_P\left(\frac{\gamma_n^2}{\Delta_{\varepsilon,J}^2}\right) \right).
\end{eqnarray*}
The result follows.~$\Box$

\end{document}